\documentclass[a4paper]{article}
\usepackage[utf8]{inputenc}
\usepackage[dvipsnames]{xcolor}
\usepackage{tikz}
\usetikzlibrary{positioning,backgrounds,patterns,calc,matrix,shapes,decorations.pathreplacing,decorations.pathmorphing,fit}
\usepackage{amsmath}
\usepackage{amsthm}
\usepackage{amsfonts}
\usepackage{amssymb}
\usepackage{xcolor,graphicx}
\usepackage{tabularx}
\usepackage[format=plain]{caption}
\usepackage[numbers,sort&compress]{natbib}
\usepackage{multirow}
\usepackage{dsfont}
\usepackage[pagebackref,menucolor=orange!40!black,filecolor=magenta!40!black,urlcolor=blue!40!black,linkcolor=red!40!black,citecolor=green!40!black,colorlinks]{hyperref}
\usepackage[sort&compress,nameinlink,noabbrev, capitalize]{cleveref}
\usepackage{booktabs}
\usepackage{enumerate}
\usepackage{a4wide}

\crefname{rrule}{Reduction Rule}{Reduction Rules}
\crefname{observation}{Observation}{Observations}

\usepackage{etoolbox}

\newcommand{\problemdef}[3]{
	\begin{center}
		\begin{minipage}{0.95\textwidth}
			\noindent
			\textsc{#1}
			\vspace{2pt}
			
			\setlength{\tabcolsep}{3pt}
			\begin{tabularx}{\textwidth}{@{}lX@{}}
				\textbf{Input:} 	& #2 \\
				\textbf{Task:} 		& #3
			\end{tabularx}
		\end{minipage}
	\end{center}
}

\newcommand{\fillRoundedNode}[6]{%
	\fill[sharp corners, fill=#2]%
		(#1.west)%
		[rounded corners=#3] |- (#1.north)%
		[rounded corners=#4] -| (#1.east)%
		[rounded corners=#5] |- (#1.south)%
		[rounded corners=#6] -| (#1.west);
}

\newcommand{\myNodeHeight}{0.5cm}
\newcommand{\myNodeWidth}{0.9cm}
\newcommand{\myNodeCornerRound}{8pt}

\newcommand{\copyPart}[3]{
			\node[
					anchor=north west,
					minimum width=\myNodeWidth,
					minimum height=\myNodeHeight,
					rectangle
				] (Left) at (Top.south west) {#1};
			\node[
					anchor=north,
					minimum width=\myNodeWidth,
					minimum height=\myNodeHeight,
					rectangle
				] (Center) at (Top.south){#2};
			\node[
					anchor=north east,
					minimum width=\myNodeWidth,
					minimum height=\myNodeHeight,
					rectangle
				] (Right) at (Top.south east) {#3};
}

\tikzset{
	pics/vhsplit/.style n args = {8}{
		code = {
			\node[
					minimum height=1.8*\myNodeHeight,
					minimum width=3*\myNodeWidth,
					rectangle,inner sep=2pt
				] (Top) at (0,0) {#1};
				
			\copyPart{#5}{#6}{#7}

			\begin{scope}%
				\fillRoundedNode{Left}{#2}{0}{0}{0}{\myNodeCornerRound}
				\fillRoundedNode{Center}{#3}{0}{0}{0}{0}
				\fillRoundedNode{Right}{#4}{0}{0}{\myNodeCornerRound}{0}
			\end{scope}

			\draw (Left.north west) -- (Right.north east);
			\draw (Center.north west) -- (Center.south west);
			\draw (Right.north west) -- (Right.south west);
			
			\copyPart{#5}{#6}{#7}
			
			\node[inner sep=0pt,draw,rounded corners=\myNodeCornerRound,fit=(Top)(Left)(Center)(Right)] () {};
		}
	}
}

\newcommand{\myTabThm}[1]{\hyperref[#1]{\textcolor{black}{{Th.\,\ref*{#1}}}}}
\newcommand{\myTabObs}[1]{\hyperref[#1]{\textcolor{black}{{Ob.\,\ref*{#1}}}}}
\newcommand{\myTabProp}[1]{\hyperref[#1]{\textcolor{black}{{Pr.\,\ref*{#1}}}}}

\newcommand{\blackCite}[1]{%
	\hypersetup{citecolor=black}%
	\textcolor{black}{\mbox{\cite{#1}}}%
	\hypersetup{citecolor=green!40!black}%
}

\newcommand{\tworows}[2]{\begin{tabular}{c}{#1}\\{#2}\end{tabular}}
\newcommand{\distto}[1]{\tworows{Distance to}{#1}}

\definecolor{colGreen}{RGB}{0, 146, 70}
\definecolor{colRed}{RGB}{206, 43, 55}
\definecolor{colViolet}{RGB}{102, 0, 204}

\tikzstyle{para}=[rectangle,draw=black,minimum height=.8cm,fill=gray!10,rounded corners=1mm, on grid]
\tikzstyle{vertex}=[circle,draw,inner sep=2pt,fill]

\newtheorem{theorem}{Theorem}
\newtheorem{lemma}{Lemma}

\newtheorem{proposition}{Proposition}
\newtheorem{observation}{Observation}
\newtheorem{definition}{Definition}
\newtheorem{rrule}{Reduction Rule}

\DeclareMathOperator*{\argmax}{arg\,max}

\DeclareMathOperator{\pen}{pen}
\DeclareMathOperator{\cl}{cl}
\DeclareMathOperator{\cc}{c-c}

\newcommand{\di}{\textsc{Diameter}}
\newcommand{\wdi}{\textsc{Weighted Diameter}}
\newcommand{\dwdi}{\textsc{Doubly Weighted Diameter}}
\newcommand{\dia}{d}%
\newcommand{\dis}{\operatorname{dist}}

\newcommand{\apsp}{\textsc{All-Pairs Shortest Paths}}
\newcommand{\param}{\textsf}
\newcommand{\dt}[1]{\param{distance to #1{}}}

\newcommand*{\defeq}{\mathrel{\vcenter{\baselineskip0.5ex \lineskiplimit0pt\hbox{\scriptsize.}\hbox{\scriptsize.}}}=}

\newcommand{\N}{\mathds{N}}
\newcommand{\w}{\ensuremath{\tau}}

\title{Parameterized Complexity of Diameter\footnote{An extended abstract appeared in the Proceedings of the 11th International Conference on Algorithms and Complexity (CIAC~'19). This version contains additional detail and omitted proofs. This work was partially supported by DFG project FPTinP, NI 369/16.}}
\author{Matthias Bentert \and Andr{\'e} Nichterlein}
\date{Algorithmics and Computational Complexity, TU~Berlin, Germany \texttt{\{matthias.bentert,andre.nichterlein\}@tu-berlin.de}}

\newcommand{\abstractCite}[1]{[\citeauthor{#1}, \citeyear{#1}]}

\begin{document}

\maketitle

\begin{abstract}
	\di---the task of computing the length of a longest shortest path---is a fundamental graph problem.
	Assuming the Strong Exponential Time Hypothesis, there is no $O(n^{1.99})$-time algorithm even in sparse graphs~\abstractCite{RW13}.
	To circumvent this lower bound, 	we investigate which parameters allow for running times of the form~$f(k) (n+m)$ where~$k$ is the respective parameter and $f$ is a computable function.
	To this end, we systematically explore a hierarchy of structural graph parameters.
\end{abstract}

\section{Introduction}
The diameter is arguably among the most fundamental graph parameters. 
Most known algorithms for determining the diameter first compute the shortest path between each pair of vertices (APSP: \apsp) and then return the maximum~\cite{AWW16}. 
The currently fastest algorithms for APSP in weighted graphs %
have a running time of $O(n^3 / 2^{\Omega(\sqrt{\log n})})$ in dense graphs~\cite{CW16} and~$O(nm + n^2 \log n)$ in sparse graphs~\cite{Joh77}, respectively.
In this work, we focus on the unweighted case.  
Formally, we study the following problem: 
\problemdef{\di}
{An undirected, connected, unweighted graph~$G=(V,E)$.}
{Compute the length of a longest shortest path in~$G$.}
The (theoretically) fastest algorithm for \di{} runs in~$O(n^{2.373})$ time and is based on fast matrix multiplication~\cite{Sei95}. 
This upper bound can (presumably) not be improved by much as \citet{RW13} showed that solving \di{} in~$O((n+m)^{2-\varepsilon})$ time for any~$\varepsilon > 0$ breaks the SETH (Strong Exponential Time Hypothesis~\cite{IP01,IPZ01}).
Seeking for ways to circumvent this lower bound, we follow the line of ``parameterization for polynomial-time solvable problems''~\cite{GMN17} (also referred to as ``FPT~in~P'').
This approach is recently actively studied and sparked a lot of research~\cite{FKMNNT19,FLSPW18,MNN17,KNNZ18,AWW16,BHM18,BFNN19,KN18,CDP18}.
Given some parameter~$k$, we aim for an algorithm with a running time of~$f(k) (n+m)$ that solves \di{}. 
Starting FPT~in~P for \di{}, \citet{AWW16} observed that, unless the SETH fails, the function~$f$ has to be an \emph{exponential} function if~$k$ is the \param{treewidth} of the graph.
We extend their research by systematically exploring the parameter space looking for parameters where~$f$ can be a polynomial. 
If such running times contradict conditional lower bounds, then we seek for matching upper bounds of the form~$f(k)(n+m)$ or~$f(k) n^2$ where~$f$ is exponential. 

In a second step, we combine parameters that are known to be small in many real-world graphs. 
We concentrate on social networks which often have special characteristics, including the ``small-world'' property and a power-law degree distribution~\cite{LH08,M67,New03,New10,NJ03}.  
We therefore combine parameters related to the \param{diameter} with parameters related to the \param{$h$-index}\footnote{The $h$-index of a graph~$G$ is the largest number~$\ell$ such that~$G$ contains at least~$\ell$ vertices of degree at least~$\ell$.}; both parameters can be expected to be orders of magnitude smaller than the number of vertices in large social networks. 

\paragraph{Related Work.}
Due to its importance, \di{} is extensively studied.
Algorithms employed in practice have usually a worst-case running time of~$O(nm)$, but are much faster in experiments.
See e.\,g. \citet{BCHKMT15} for a recent example which also yields good performance bounds using average-case analysis~\cite{BCT17}.
Concerning worst-case analysis, the theoretically fastest algorithms are based on matrix multiplication and run in~$O(n^{2.373})$ time~\cite{Sei95} and any $O((n+m)^{2-\varepsilon})$-time algorithm refutes the SETH~\cite{RW13}.

The following results on approximating \di{} are known:
It is easy to see that a simple breadth-first search gives a linear-time $2$-approximation.
\citet{ACIM99} improved the approximation factor to~$3/2$ at the expense of the higher running time of~$O(n^2 \log n + m \sqrt{n \log n})$.
The lower bound of \citet{RW13} also implies that approximating \di{} within a factor of~$3/2 - \delta$ in~$O(n^{2 - \varepsilon})$ time refutes the SETH.
Moreover, for any~$\varepsilon,\delta > 0$ a $(3/2-\delta)$-approximation in $O(m^{2-\varepsilon})$ time or a $(5/3-\delta)$-approximation in $O(m^{3/2-\varepsilon})$ time also refute the SETH~\cite{CGR16,BRSWW18}.
On planar graphs, there is an approximation scheme with %
near linear running time~\cite{WY16}; the fastest exact algorithm for \di{} on planar graphs runs in~$O(n^{1.667})$ time~\cite{GKMS18}.

Concerning FPT~in~P, \di{} can be solved in~$2^{O(k)} n^{1+o(1)}$ time where~$k$ is the \param{treewidth} of the graph~\cite{BHM18}.
However, the reduction for the lower bound of \citet{RW13} implies that for any~$\varepsilon > 0$ a $2^{o(k)} n^{2-\varepsilon}$-time algorithm refutes the SETH, where~$k$ is either the \param{vertex cover number}, the \param{treewidth}, or the combined parameter \param{$h$-index and domination number}.
Moreover, this reduction also implies that the SETH is refuted by any~$f(k)(n+m)^{2-\varepsilon}$-time algorithm for \di{} for any computable function~$f$ and~$\varepsilon > 0$ when~$k$ is the \param{(vertex deletion) distance to chordal graphs}.
\citet{ED16} adapted the reduction by \citeauthor{RW13} and proved that any~$f(k)(n+m)^{2-\varepsilon}$-time algorithm for \di{} parameterized by the \param{maximum degree}~$k$ for any computable function~$f$ refutes the SETH.

\paragraph{Our Contribution.}
We make progress towards systematically classifying the complexity of \di{} parameterized by structural graph parameters. 
\Cref{fig:param-hierarchy} gives an overview of previously known and new results and their implications.
\begin{figure}[t!]
	\newcommand{\ColorOpen}{white}
	\newcommand{\ColorHard}{red!56}%
	\newcommand{\ColorFPT}{green!50}%
	\resizebox{\textwidth}{!}
	{
		\begin{tikzpicture}
			\matrix[ampersand replacement=\&,row sep=0.75cm,column sep=0.25cm]
			{
			
				\& 
				\& 
				\& \pic (md-ds) {vhsplit={\tworows{Max. Degree}{+ Dom. No.}}{\ColorFPT}{\ColorFPT}{\ColorFPT}{ }{ }{\myTabObs{obs:maxdeg-ds-alg}}};
				\& 
				\& \\
				\& \pic (vc) {vhsplit={\tworows{Vertex Cover}{Number}}{\ColorFPT}{\ColorFPT}{\ColorHard}{ }{ }{\blackCite{RW13}}};
				\& 
				\& \pic (h-ds) {vhsplit={\tworows{$h$-index + }{Dom. No.}}{\ColorFPT}{\ColorFPT}{\ColorHard}{ }{ }{\blackCite{RW13}}};
				\& \pic (md-diam) {vhsplit={\tworows{Max Degree}{+ Diameter}}{\ColorFPT}{\ColorFPT}{\ColorHard}{ }{ }{ \myTabThm{thm:maxdeg-diam} }};
				\& \pic (bsw) {vhsplit={\tworows{Bisection}{Width}}{\ColorHard}{\ColorHard}{\ColorHard}{\myTabThm{thm:BWisGPhard}}{ }{ }}; \\
				\pic (dc) {vhsplit={\distto{Clique}}{\ColorFPT}{\ColorFPT}{\ColorFPT}{ }{ }{ \myTabObs{obs:dist-to-clique} }};
				\& 
				\& \pic (fes) {vhsplit={\tworows{Feedback}{Edge Number}}{\ColorFPT}{\ColorFPT}{\ColorFPT}{ }{ }{\myTabThm{thm:fes}}};
				\& \pic (acn-ds) {vhsplit={\tworows{Ac. Chrom. No.}{+ Dom. No.}}{\ColorFPT}{\ColorHard}{\ColorHard}{ }{\myTabThm{thm:acn-ds}}{}};
				\& \pic (h-diam) {vhsplit={\tworows{$h$-index + }{Diameter}}{\ColorFPT}{\ColorFPT}{\ColorHard}{ }{ \myTabThm{thm:hind-diam} }{ }};
				\& \pic (max) {vhsplit={\tworows{Maximum}{Degree}}{\ColorFPT}{\ColorHard}{\ColorHard}{ }{\blackCite{ED16}}{ }}; \\
				\& \pic (dig) {vhsplit={\distto{Interval}}{\ColorFPT}{\ColorOpen}{\ColorHard}{ \myTabObs{obs:dist-to-interval} }{ }{ }};
				\& 
				\& \pic (tw) {vhsplit={\tworows{Treewidth}{}}{\ColorFPT}{\ColorFPT}{\ColorHard}{ }{\blackCite{AWW16,BHM18}}{ }};
				\& \pic (acn-diam) {vhsplit={\tworows{Ac. Chrom. No.}{ + Diameter}}{\ColorFPT}{\ColorHard}{\ColorHard}{ }{ }{ }};
				\& \pic (hind) {vhsplit={\tworows{$h$-index}{}}{\ColorFPT}{\ColorHard}{\ColorHard}{ }{ }{ }};\\
				\pic (ds) {vhsplit={\tworows{Domination}{Number}}{\ColorOpen}{\ColorHard}{\ColorHard}{ }{\blackCite{RW13}}{}};
				\& \pic (dcg) {vhsplit={\distto{Cograph}}{\ColorFPT}{\ColorFPT}{\ColorHard}{ }{\myTabThm{thm:cograph}}{ }};
				\& \pic (dco) {vhsplit={\distto{Chordal}}{\ColorOpen}{\ColorHard}{\ColorHard}{ }{\blackCite{RW13}}{ }};
				\& 
				\& 
				\& \pic (acn) {vhsplit={\tworows{Acyclic}{Chromatic No.}}{\ColorFPT}{\ColorHard}{\ColorHard}{ }{ }{ }};\\
				\pic (diam) {vhsplit={\tworows{Max Diameter}{of Components}}{\ColorOpen}{\ColorHard}{\ColorHard}{ }{ }{ }};
				\& 
				\& \pic (dbg) {vhsplit={\distto{Bipartite}}{\ColorHard}{\ColorHard}{\ColorHard}{\myTabThm{thm:DtBisGPhard}}{ }{ }};
				\& 
				\& 
				\& \pic (avg) {vhsplit={\tworows{Average}{Degree}}{\ColorFPT}{\ColorHard}{\ColorHard}{\myTabObs{obs:avg-degree}}{ }{ }};\\
				\pic (girth) {vhsplit={\tworows{Girth}{}}{\ColorHard}{\ColorHard}{\ColorHard}{\myTabThm{thm:DtBisGPhard}}{ }{ }};
				\& \pic (dpg) {vhsplit={\distto{Perfect}}{\ColorHard}{\ColorHard}{\ColorHard}{ }{ }{ }};
				\& 
				\& \pic (cn) {vhsplit={\tworows{Chromatic}{Number}}{\ColorHard}{\ColorHard}{\ColorHard}{ }{ }{ }};
				\& 
				\& \pic (min) {vhsplit={\tworows{Minimum}{Degree}}{\ColorHard}{\ColorHard}{\ColorHard}{\myTabThm{thm:BWisGPhard}}{ }{ }};\\
			};
		\draw[thick] (md-ds) -- (md-diam);
		\draw[thick] (h-diam) -- (acn-diam);
		\draw[thick] (acn-diam) -- (acn)  -- (avg) -- (min);
		\draw[thick] (md-ds) -- (h-ds) -- (acn-ds) -- (acn-diam);;
		\draw[thick] (diam) to[out=15,in=230] (acn-diam);
		\draw[thick] (md-ds) to[out=340,in=140] (bsw);
		\draw[thick] (md-ds) -- (vc) to[out=340,in=155] (h-diam);
		\draw[thick] (md-ds) to[out=180,in=80] (dc);
		\draw[thick] (md-ds) to[out=205,in=85] (fes);
		\draw[thick] (md-diam) -- (max);
		\draw[thick] (md-diam) -- (h-diam);
		\draw[thick] (h-ds) -- (h-diam);
		\draw[thick] (h-diam) -- (hind);
		\draw[thick] (hind)-- (acn);
		\draw[thick] (max) -- (hind);
		\draw[thick] (vc) -- (dig);
		\draw[thick] (vc) to[out=290,in=175] (tw);
		\draw[thick] (vc) to[out=285,in=32] (dcg);
		\draw[thick] (dc) -- 
				(ds) -- (diam) -- (girth);
		\draw[thick] (dc) to[out=280,in=145] (dcg);
		\draw[thick] (dc) -- (dig);
		\draw[thick] (dig)-- (dco);
		\draw[thick] (dco) to[out=220,in=80] (dpg);
		\draw[thick] (dcg) -- (dpg);
		\draw[thick] (dcg) -- (diam);
		\draw[thick] (dbg) -- (cn);
		\draw[thick] (dbg) -- (dpg);
		\draw[thick] (fes) -- (tw);
		\draw[thick] (fes) to[out=290,in=35] (dbg);
		\draw[thick] (fes) -- (dco);
		\draw[thick] (tw) to[out=327,in=170] (acn);
		\draw[thick] (acn) -- (cn);
		
	\end{tikzpicture}
	}
	\caption{
		Overview of the relation between the structural parameters and the respective results for \di{}.
		An edge from a parameter~$\alpha$ to a parameter~$\beta$ below of~$\alpha$ means that~$\beta$ can be upper-bounded in a polynomial (usually linear) function in~$\alpha$ (see also~\cite{SW13}). 
		The three small boxes below each parameter indicate whether there exists (from left to right) an algorithm running in $f(k)n^2$, $f(k)(n \log n +m)$, or~$k^{O(1)}(n \log n+m)$ time, respectively.
		If a small box is green, then a corresponding algorithm exists and the box to the left is also green.
		Similarly, a red box indicates that a corresponding algorithm is a breakthrough. 
		More precisely, if a middle box (right box) is red, then an algorithm running in~$f(k) \cdot (n+m)^{2-\varepsilon}$ (or~$k^{O(1)} \cdot (n+m)^{2-\varepsilon}$) time refutes the SETH.
		If a left box is red, then an algorithm with running time~$f(k)n^2$ implies an~$O(n^2)$ time algorithm for \di{} in general.
		Hardness results for a parameter~$\alpha$ imply the same hardness results for the parameters below~$\alpha$.
		Similarly, algorithms for a parameter~$\beta$ imply algorithms for the parameters above~$\beta$. 
		We remark that in the above hierarchy only the algorithm behind the green box for the parameter \param{distance to interval} requires additional input related to the parameter (here the modulator to an interval graph).
	}
	\label{fig:param-hierarchy}
\end{figure}
We define the graph parameters for which we provide results in the sections where they are used; we refer to \citet{BLS99} for definitions of the remaining parameters in \Cref{fig:param-hierarchy}. 

In \cref{sec:dist-from-class}, we follow the ``distance from triviality parameterization'' \cite{GHN04} aiming to extend known tractability results for special graph classes to graphs with small modulators.
For example, \di{} is linear-time solvable on trees.
We obtain an~$O(k \cdot n)$-time algorithm for the parameter \param{feedback edge number~$k$} (edge deletion number to trees). 
However, this is our only~$k^{O(1)}(n+m)$-time algorithm in this section. 
For the remaining parameters, it is already known that such algorithms refute the SETH.
For the parameter \param{distance~$k$ to cographs} we therefore provide a $2^{O(k)}(n+m)$-time algorithm.
Finally, for the parameter \param{odd cycle transversal}~$k$, we use the recently introduced notion of \emph{General-Problem-hardness}~\cite{BFNN19} to show that \di{} parameterized by~$k$ is ``as hard'' as the unparameterized \di{} problem.
In \cref{sec:soc-networks}, we investigate parameter combinations.
We prove that a~$k^{O(1)}\allowbreak(n+m)^{2-\varepsilon}$-time algorithm where~$k$ is the combined parameter \param{diameter and maximum degree} would refute the SETH.
Complementing this lower bound, we provide an~$f(k)(n+m)$-time algorithm where~$k$ is the combined parameter \param{diameter and $h$-index}.

Many of our algorithmic results for \di{} transfer easily to the edge-weighted case by simply exchanging bread-first search with Dijkstra's algorithm and thus getting a logarithmic overhead in the running time.
Whenever this is the case, we state the result for the edge-weighted case which we call \wdi.
The focus of our work (and hence the overview in \Cref{fig:param-hierarchy}) is still on the unweighted case.
Thus, we provide hardness results only for the easier, unweighted variant \di{}.

\section{Preliminaries}%
We set~$\N := \{0,1,2,\ldots,\}$ and~$\N^+ := \N \setminus \{0\}$.
For $\ell \in \N^+$ we set $[\ell] := \{1, 2,\ldots, \ell\}$.
We use mostly standard graph notation.
For a graph~$G = (V,E)$ we set~$n:=|V|$ and~$m:= |E|$.
All graphs in this work are undirected.
For a vertex subset~$V' \subseteq V$, we denote with~$G[V']$ the graph induced by~$V'$.
We set~$G-V' := G[V \setminus V']$.
A path~$P = v_0 \dots v_a$ is a graph with vertex set~$\{v_0, \ldots, v_a\}$ and edge set~$\{\{v_i,v_{i+1}\} \mid 0 \le i < a \}$. 
For~$u,v \in V$, we denote with~$\dis_G(u,v)$ the distance between~$u$ and~$v$ in~$G$, that is, the number of edges (the sum of edge weights in weighted graphs) in a shortest path between~$u$ and~$v$.
If~$G$ is clear from the context, then we omit the subscript.
We denote by~$\dia(G)$ the diameter of~$G$, that is, the length of the longest shortest path in~$G$.
For \wdi{} we consider edge weights to be positive integers:
\problemdef{\wdi}
{A connected graph~$G=(V,E)$ and edge weights~$\w\colon E \rightarrow \N^+$.}
{Compute~$\dia(G)$.}

\paragraph{Parameterized Complexity and GP-hardness.}
A language~$L\subseteq \Sigma^* \times \N$ is a \emph{parameterized problem} over some finite alphabet~$\Sigma$, where~$(x,k) \in \Sigma^* \times \N$ denotes an instance of~$L$ and~$k$ is the parameter.
The language~$L$~is called \emph{fixed-parameter tractable} if there is an algorithm that on input~$(x,k)$ decides whether~$(x,k)\in L$ in~$f(k)\cdot |x|^{O(1)}$ time, where~$f$ is some computable function only depending on~$k$ and~$|x|$ denotes the size of~$x$.
For a parameterized problem~$L$, the language~$\hat{L}=\{x\in \Sigma^*\mid \exists k\colon (x,k)\in L\}$ is called the \emph{unparameterized problem} associated to~$L$.
We use the notion of General-Problem-hardness which formalizes the types of reduction that allow us to exclude parameterized algorithms as they would lead to faster algorithms for the general, unparameterized, problem.
\begin{definition}[{\cite[Definition 2]{BFNN19}}]
	\label{def:k-GP-hard}
	Let~$P \subseteq \Sigma^* \times \N$ be a parameterized problem, let~$\hat{P} \subseteq \Sigma^*$ be the unparameterized decision problem associated to~$P$, and let~$g\colon \N \rightarrow \N$ be a polynomial.
	We call~$P$~\emph{$\ell$-General-Problem-hard$(g)$ ($\ell$-GP-hard$(g)$)} if there exists an algorithm~$\cal{A}$ transforming any input instance~$I$ of~$\hat{P}$ into a new instance~$(I',k')$ of~$P$ such that
	\begin{enumerate}[({G}1)]
		\item~$\cal{A}$ runs in~$O(g(|I|))$ time,\label{prop:GPh-time}
		\item~$I \in \hat{P} \iff (I',k') \in P$,\label{prop:GPh-equiv}
		\item~$k' \leq \ell$, and\label{prop:GPh-param}
		\item~$|I'| \in O(|I|)$.\label{prop:GPh-size}
	\end{enumerate}%
	We call~$P$ \emph{General-Problem-hard$(g)$ (GP-hard$(g)$)} if there exists an integer~$\ell$ such that~$P$ is~$\ell$-GP-hard$(g)$. 
	We omit the running time and call~$P$ \emph{$\ell$-General-Problem-hard ($\ell$-GP-hard)} if~$g$ is a linear function. 
\end{definition}
Showing GP-hardness for some parameter~$k$ allows to lift algorithms for the parameterized problem to the unparameterized setting as stated next.
\begin{lemma}[{\cite[Lemma 3]{BFNN19}}]
	\label{lem:gph}
	Let~$g\colon \N \rightarrow \N$ be a polynomial, let~$P \subseteq \Sigma^* \times \N$ be a parameterized problem that is~GP-hard$(g)$, and let~$\hat{P} \subseteq \Sigma^*$ be the unparameterized decision problem associated to~$P$.
	If there is an algorithm solving each instance~$(I, k)$ of~$P$ in~$O(f(k) \cdot g(|I|))$ time, then there is an algorithm solving each instance~$I'$ of~$\hat{P}$ in~$O(g(|I'|))$ time. 
\end{lemma}

Applying \cref{lem:gph} to \di{} yields the following.
First, having an $f(k) n^{2.3}$ time algorithm with respect to a parameter~$k$ for which \di{} is GP-hard would yield a faster \di{} algorithm.
Moreover, from the known SETH-based hardness results~\cite{RW13,CGR16,BRSWW18}, we get the following.

\begin{observation}
	If the SETH is true and \di{} is GP-hard($n^{2-\varepsilon}$) with respect to some parameter~$k$ for some~$\varepsilon > 0$, then there is no $f(k) \cdot n^{2-\varepsilon'}$ time algorithm for any~$\varepsilon' > 0$ and any function~$f$.
\end{observation}

\section{Basic Observations}
In this section, we present several simple observations that complete the overview in \cref{fig:param-hierarchy}.
More precisely, we show algorithms with respect to the parameters \param{distance~$c$ to clique}, \param{distance~$i$ to interval graphs}, \param{average degree~$a$},  \param{maximum degree}~$\Delta$, \param{diameter}~$\dia$, and \param{domination number}~$\gamma$ (in the order they are listed).

\paragraph{Distance to clique.}
We start with the parameter \param{distance~$c$ to clique} and provide an algorithm with running time~$O(c \cdot (n + m))$ time.
Since \dt{clique} is the \param{vertex cover number} in the complement graph, it can be 2-approximated in linear time (without computing the complement graph). %

\begin{observation}
	\label{obs:dist-to-clique}
	\di{} parameterized by \param{distance~$c$ to clique} takes~$O(c \cdot (n + m))$ time.
\end{observation}

\begin{proof}
Let~$G=(V,E)$ be the input graph and let~$c$ be its \param{distance to clique}.
Let~$G'$ be the respective induced clique graph. 
Compute in linear time the degree of each vertex and the number~$n = |V|$ of vertices.
Iteratively check for each vertex~$v$ whether its degree is~$n-1$.
If~$\deg(v) = n-1$, then~$v$ can be deleted as it is in every largest clique and thus decrease~$n$ by one and the degree of each other vertex by one.
If not, then we can find a vertex~$w$ which is not adjacent to~$v$ in~$O(\deg(v))$ time.
Put~$v$ and~$w$ in the solution set, delete both vertices and all incident edges and adjust the number of vertices and their degree accordingly.
Observe that~$v$ and~$w$ cannot be contained in the same clique and therefore~$v\in K$ or~$w\in K$.
Putting both vertices in the solution set results in a 2-approximation.
This algorithm takes~$O(\deg(v) + \deg(w))$ time per deleted pair~$v,w$ of vertices.
Since~$\sum_{v\in V} \deg(v) \in O(n+m)$ this procedure takes~$O(n+m)$ time.

We use the algorithm described above to compute a set~$K$ such that~$G' = G-K$ is a clique and~$|K| \leq 2k$ in linear time.
Since~$G'$ is a clique, its diameter is one if there are at least two vertices in the clique.
We therefore assume that there is at least one vertex in the deletion set~$K$.
Compute for each vertex~$v\in K$ a breadth-first search rooted in~$v$ in linear time and return the largest distance found.
The returned value is the diameter of~$G$ as each longest induced path is either of length one or has at least one endpoint in~$K$.
The procedure takes~${O(|K|\cdot (n+m) + n + m) = O(c\cdot (n+m))}$ time.
\end{proof}

Note that for \wdi{} a result similar to \cref{obs:dist-to-clique} would yield a faster algorithm for \di:
In a clique~$C$ with~$n$ vertices and edge weights either~$1$ or~$n$, one can encode any connected unweighted graph~$G$ by giving edges in~$G$ weight one in~$C$ and any non-edge in~$G$ a weight of~$n$ in~$C$. 
It is easy to see that~$G$ has the same diameter as~$C$.
Thus, an algorithm for \wdi{} with running time $O(c \cdot (n + m))$ would imply an~$O(n^2)$ algorithm for \di and, hence, drastically improve on the state-of-the-art.

\paragraph{Distance to interval graphs.}
We next discuss the parameter \dt{interval graphs}.
We first provide a general observation stating that a size~$k$ deletion set to some graph class can be used to design a $O(k\cdot n^2)$-time algorithm if \apsp{} can be solved in~$O(n^2)$ time on graphs in the respective graph class.
The algorithm is fairly simple:
First compute~$G'$, that is, the graph without the deletion set~$K$, and solve \apsp{} on it in~$O(n^2)$ time.
Next, compute a breadth-first search from every vertex in~$K$ in~${O(k \cdot m)}$ time and store all distances found in a table.
The last step can  be seen as running the classical Floyd-Warshall algorithm for each vertex in~$K$: compute for each pair~$a,c\in V\setminus K$ 
$$\dis_G(a,c) := \min \{\dis_{G'}(a,c), \min_{b\in K}\{\dis_G(a,b)+\dis_G(b,c)\}\},$$
that is, the minimum distance in the original graph.
Observe that a shortest path either travels through some vertex $b\in K$ or not.
In the latter case,~$\dis_{G}(a,c) = \dis_{G'}(a,c)$ and in the former case the distance between~$a$ and~$c$ in~$G$ is~$\dis_G(a,b)+\dis_G(b,c)$.
This algorithm takes~$O(n+m+n^2+k\cdot m+n^2\cdot k) = O(k\cdot n^2)$ time.

\begin{observation}
	\label{obs:dist-to-graphclass}
	Let~$\Pi$ be a graph class such that \apsp{} can be solved in~$O(n^2)$ time on~$\Pi$.
	If the (vertex) deletion set~$K$ to~$\Pi$ is given, then \apsp{} can be solved in~\mbox{$O(|K|\cdot n^2)$} time.
\end{observation}
Note the above algorithm works also for weighted graphs by replacing the breadth-first search with Dijkstra's algorithm.
The overall running time would be unchanged as running Dijkstra's algorithm~$k$ times takes~$O(kn\log n + km) = O(kn^2)$ time.
Thus, \cref{obs:dist-to-graphclass} is true for weighted and unweighted graphs.

It is known that (unweighted) \apsp{} can be solved in~$O(n^2)$ time on interval graphs~\cite{CLSS98,ST99}.
Thus we obtain the following.

\begin{observation}
	\label{obs:dist-to-interval}
	\di{} parameterized by the \param{distance~$i$ to interval graphs} is solvable in~\mbox{$O(i\cdot n^2)$} time provided that the deletion set 
	is given.
\end{observation}

We are not aware of a fast constant factor approximation algorithm to compute the deletion set in the above observation.
Finding (or excluding) such an approximation algorithm remains a task for future work.
As interval graphs contain cliques, it follows again that generalizing \cref{obs:dist-to-interval} to the weighted case would improve upon the state-of-the-art algorithm for \di{}.

\paragraph{Average degree.}
We next consider the \param{average degree}~$a$.
Observe that~$2m = n \cdot a$ and therefore the standard algorithm (run Dijkstra's algorithm~$n$ times) takes~$O(n\cdot (n \log n + m)) = O(n^2 (\log n + a))$ time.

\begin{observation}
	\label{obs:avg-degree}
	\wdi{} parameterized by \param{average degree}~$a$ is solvable in~$O((a + \log n) \cdot n^2)$ time.
\end{observation}

\paragraph{Maximum degree and diameter.}
We look at two parameter combinations related to both \param{maximum degree} and \param{diameter}.
Usually, this parameter is not interesting as the graph size can be upper-bounded by this parameter and thus fixed-parameter tractability with respect to this combined parameter is trivial.
The input size is, however, only exponentially bounded in the parameter, so it might be tempting to search for fully polynomial algorithms.
In \cref{sec:comb} we exclude such a fully polynomial algorithm.
Thus, the subsequent algorithm is basically optimal.

\begin{observation}
	\label{obs:maxdeg-diam-alg}
	\wdi{} parameterized by~\param{diameter}~$d$ and \param{maximum degree}~$\Delta$ is solvable in~$O(\Delta^{2d} \cdot (d \log \Delta + \Delta))$ time.
\end{observation}
\begin{proof}
Since we may assume that the input graph only consists of one connected component, every vertex is found by any breadth-first search.
Any breadth-first search may only reach depth~$d$, where~$d$ is the diameter of the input graph, and as each vertex may only have~$\Delta$ neighbors there are at most~$1+\sum_{i=1}^{d} \Delta \cdot (\Delta-1)^{i-1} \leq 1+\sum_{i=1}^{d} \Delta^{i-1} \cdot (\Delta-1)  = \Delta^{d}$ vertices (since in each ``depth layer~$i$'' there are at most~$\Delta \cdot (\Delta-1)^{i-1}$ vertices).
Since $m \leq n\cdot \Delta$ the $O(n \cdot (n \log n + m))$-time algorithm ($n$ rounds of Dijkstra's algorithm) runs in~$O(\Delta^{2d} \cdot (d \log \Delta + \Delta))$ time. 
\end{proof}

\paragraph{Maximum degree and domination number.}
Observe that for any graph of~$n$ vertices, domination number~$\gamma$, and maximum degree~$\Delta$ it holds that~$n\leq \gamma \cdot (\Delta+1)$ as each vertex is in a dominating set or is a neighbor of at least one vertex in it.
The next observation follows from~$m \leq n\cdot \Delta$.
\begin{observation}
	\label{obs:maxdeg-ds-alg}
	\wdi{} parameterized by \param{domination number}~$\gamma$ and \param{maximum degree}~$\Delta$ is solvable in~$O(\gamma^2 \Delta^2 (\Delta + \log(\gamma\Delta)))$ time. 
\end{observation}

The reduction of \citet{RW13} can also be used to show that the SETH is refuted by any~$f(\gamma)(n+m)^{2-\varepsilon}$-time algorithm for \di{} for any computable function~$f$ even if a minimum dominating set is given.
This lower bound result is in stark contrast to a simple algorithm running in~$O(\gamma(n+m))$ time that returns either the diameter or the diameter minus one.

\begin{observation}
	Given a dominating set of size~$\gamma$ for an unweighted graph, one can approximate the diameter with an additive factor of one in~$O(\gamma(n+m))$ time.
\end{observation}
\begin{proof}
	The algorithm is as follows: Run a breadth-first search from each vertex in the dominating set~$D$ and return the largest distance found.
	This can be done~$O(\gamma(n+m))$ time.
	Clearly, the value~$\ell$ returned by the algorithm is at most the diameter~$d$ of the input graph, that is, $\ell \le d$.
	It remains to show that~$d \le \ell + 1$.
	
	To this end, let~$u,v$ be the two furthest vertices, that is, ${\dis(u,v) = d}$.
	Observe that if either~$u$ or~$v$ is in the dominating set~$D$, then the algorithm returned~$\ell = d$.
	Thus, consider the case that neither~$u$ nor~$v$ are in~$D$.
	Since~$D$ is a dominating set, there is a vertex~$w \in D$ that is a neighbor of~$u$.
	Since~$w \in D$, the returned value is at least~$\ell \ge \dis(w,v)$.
	Hence, we have $d = \dis(u,v) \le \dis(w,v) + 1 \le \ell + 1$.
\end{proof}

Note that, although computing a minimum dominating set is NP-hard, a simple greedy algorithm computes a ($1 + \log n$)-approximation.
Thus, if the dominating set is not given, the worst-case running time of the above plus-one-approximation changes to~$O(\gamma(n+m)\log n)$, which is still far better than the lower bound for exactly computing the diameter.

\section{Deletion Distance to Special Graph Classes}
\label{sec:dist-from-class}

In this section, we investigate parameterizations that measure the distance to special graph classes.
The hope is that when \di{} can be solved efficiently in a special graph class~$\Pi$, then \di{} can be solved if the input graph is ``almost'' in~$\Pi$.
We study the following parameters in this order: \param{odd cycle transversal} (which is the same as \dt{bipartite graphs}), \dt{cographs}, and \param{feedback edge number}.
The first two parameters measure the vertex deletion distance to bipartite graphs and cographs, respectively. 
\param{Feedback edge number} measures the edge deletion distance to trees.
Note that the lower bound of \mbox{\citet{AWW16}} for the parameter \param{vertex cover number} (i.\,e.\ vertex deletion to edgeless graphs) already implies that there is no~$2^{o(k)}(n+m)^{2-\varepsilon}$-time algorithm for~$k$ being one of the first two parameters in our list unless the SETH breaks, since each of these parameters is smaller than the \param{vertex cover number} (see \cref{fig:param-hierarchy}).

\paragraph{Odd Cycle Transversal.}
We show that \di{} parameterized by \param{odd cycle transversal} and \param{girth} is~$4$-GP-hard.
Consequently, solving \di{} in $f(k) \cdot n^{2.3}$ for any computable function $f$ implies an $O(n^{2.3})$-time algorithm for \di{}---which would improve the currently best (unparameterized) algorithm.
The \param{girth} of a graph is the length of a shortest cycle in it.%
\begin{theorem}
	\label{thm:DtBisGPhard}
	\di{} is $4$-GP-hard with respect to the combined parameter \param{odd cycle transversal and girth}.
\end{theorem}

\begin{proof}
	Let $G=(V,E)$ be an arbitrary undirected graph where $V = \{v_1,v_2,\ldots, \allowbreak v_n\}$.
	We construct a new graph~$G'=(V',E')$ as follows:
	$V' \defeq \{u_i, w_i \mid v_i \in V\}\text{, and}$
	$E' \defeq \{\{u_i, w_j\}, \{u_j,w_i\} \mid \{v_i, v_j\} \in E\} \cup \{\{u_i, w_i\}\mid v_i\in V\}.$

	An example of this construction can be seen in \cref{fig:bipartitegp}.
	\begin{figure}[t!]
		\centering
		\begin{tikzpicture}[]
			\begin{scope}[xshift = -1cm]
				\node[vertex,label=above:$v_1$] (v1) at (0,0) {};
				\node[vertex,label=above:$v_2$] (v2) at (2,0) {} edge (v1);
				\node[vertex,label=right:$v_3$] (v3) at (1,-1.5) {} edge(v1) edge (v2);
				\node[vertex,label=right:$v_4$] (v4) at (1,-3) {} edge (v3);
			\end{scope}
			
			\draw [->,decorate,decoration=snake] (1.85,-1.5) -- (4,-1.5);
			
			\node[vertex,label=above:$u_1$] (u1) at (5,-1) {};
			\node[vertex,label=above:$u_2$] (u2) at (8,0) {};
			\node[vertex,label=right:$u_3$] (u3) at (8,-2) {};
			\node[vertex,label=left:$u_4$] (u4) at (6,-3) {};
			\node[vertex,label=above:$w_1$] (w1) at (6,0) {} edge (u1) edge (u2) edge (u3);
			\node[vertex,label=above:$w_2$] (w2) at (9,-1) {} edge (u1) edge (u2)  edge (u3);
			\node[vertex,label=left:$w_3$] (w3) at (6,-2) {} edge (u1) edge (u2) edge (u3)  edge (u4);
			\node[vertex,label=right:$w_4$] (w4) at (8,-3) {} edge (u3) edge (u4);
			
			\tikzstyle{edge} = [color=black,opacity=.15,line cap=round, line join=round, line width=12pt]
			\begin{pgfonlayer}{background}
				\draw[edge] (v1.center) -- (v3.center) -- (v4.center);

				\draw[edge] (w1.center) -- (u3.center) -- (w4.center) -- (u4.center);
			\end{pgfonlayer}
		\end{tikzpicture}
		\caption{
				Example for the construction in the proof of \Cref{thm:DtBisGPhard}. 
				The input graph given on the left side has diameter two and the constructed graph on the right side has diameter three.
				In each graph one longest shortest path is highlighted.
			}
		\label{fig:bipartitegp}
	\end{figure}
	We will now prove that all properties of \cref{def:k-GP-hard} hold.
	It is easy to verify that the reduction can be implemented in linear time and therefore the resulting instance is of linear size as well.
	Observe that~$\{u_i \mid v_i \in V\}$ and~$\{w_i \mid v_i \in V\}$ are both independent sets and therefore~$G'$ is bipartite.
	Notice further that for any edge~$\{v_i,v_j\}\in E$ there is an induced cycle in~$G'$ containing the vertices~$\{u_i,w_i,u_j,w_j\}$.
	Since~$G'$ is bipartite there is no induced cycle of length three in~$G'$ and thus the girth of~$G'$ is four.

	Lastly, we show that~$\dia(G') = \dia(G)+1$ by proving that if~$\dis(v_i,v_j)$ is odd, then~$\dis(u_i,w_j) = \dis(v_i,v_j)$ and~$\dis(u_i,u_j) = \dis(v_i,v_j) + 1$, and if~$\dis(v_i,v_j)$ is even, then~$\dis(u_i,u_j) = \dis(v_i,v_j)$ and~$\dis(u_i,w_j) = \dis(v_i,v_j) + 1$.
	Since $\dis(u_i,w_i) = 1$ and~${\dis(u_i,w_j) = \dis(u_j,w_i)}$, this will conclude the proof.

	Let $P=v_{a_0}v_{a_1} \ldots v_{a_d}$ be a shortest path from~$v_i$ to~$v_j$ where~$v_{a_0}=v_i$ and~$v_{a_d} = v_j$.
	Let~$P'=u_{a_0}w_{a_1}u_{a_2} w_{a_3}\ldots$ be a path in~$G'$.
	Clearly, $P'$ is also a shortest path as there are no edges~$\{u_i, w_j\}\in E'$ where~$\{v_i,v_j\}\notin E$.

	If~$d$ is odd, then~$u_{a_0}w_{a_1}\ldots w_{a_d}$ is a path of length~$d$ from~$u_i$ to~$w_j$ and~$u_{a_0}w_{a_1}\allowbreak\ldots w_{a_d} u_{a_d}$ is a path of length~$d+1$ from~$u_i$ to~$u_j$.
	If~$d$ is even, then~$u_{a_0}w_{a_1}\ldots \allowbreak w_{a_{d-1}}u_{a_d}$ is a path of length~$d$ from~$u_i$ to~$u_j$ and~$u_{a_0} w_{a_1}\ldots w_{a_{d-1}}u_{a_d}w_{a_d}$ is a path of length~$d+1$ from~$u_i$ to~$w_j$. 
	Notice that~$G'$ is bipartite and thus~$\dis(u_i, u_j)$ must be even and~$\dis(u_i,w_j)$ must be odd.
\end{proof}

\paragraph{Distance to cographs.}
A graph is a cograph if and only if it does not contain a $P_4$ as an induced subgraph, where a~$P_4$ is a path on four vertices.
Providing an algorithm that matches the lower bound of \citet{AWW16}, we will show that \di{} parameterized by \param{distance~$k$ to cographs} can be solved in~$O(k \cdot (n + m) + 2^{O(k)})$ time.
To this end, we will use the following lemma covering the algorithm in a more general setting than we use.

\begin{lemma}
	\label{lem:dist-to-2-club}
	Let~$G = (V,E)$ be an edge-weighted graph and let~$K \subseteq V$ a vertex subset such that each connected component in~$G-K$ has diameter at most two.
	Then, the diameter of~$G$ can be computed in~$O(k \cdot (n \log n + m + 2^{4k}))$ time.
\end{lemma}

\begin{proof}
	We first compute all connected components and their diameter in~$G' := G-K$ in linear time and store for each vertex the information in which connected component it is.
	Note that we only need to check for each connected component~$C$, whether~$C$ induces a clique in~$G'$ and all edge-weights are one in~$C$; otherwise $C$'s diameter is by assumption two.
	In a second step, we perform in~$O(k\cdot (n\log n + m))$ time Dijkstra's algorithm in~$G$ from each vertex $v \in K$ and store the distance between~$v$ and every other vertex~$w$ in a table.

	Next we introduce some notation.
	The type of a vertex~$u\in V\setminus K$ is a vector of length~$d$ where the~$i$th entry describes the distance from~$u$ to~$x_i$ with the addition that any value above three is set to~$4$.
	We say a type is non-empty, if there is at least one vertex with this type.
	We compute for each vertex~$u\in V\setminus K$ its type.
	Additionally we store for each non-empty type the connected component its vertex is in or that there are at least two different connected components containing a vertex of that type.
	This takes~$O(n \cdot k)$ time and there are at most~$4^k$ many different types.

	Lastly, we iterate over all of the~$O(4^{2k})$ pairs of types (including the pairs where both types are the same) and compute the largest distance between vertices of these types.
	Let~$y,z$ be the vertices of the respective types with maximum pairwise distance.
	We will first discuss how to find~$y$ and~$z$ and then show how to correctly compute their distance in~$O(k)$ time.
	If both types only appear in the same connected component, then the distance between the two vertices of these types is at most two.
	Hence, we can discard this case (one can check in linear time whether the diameter of~$G$ is at least two).
	If two types appear in different connected components, then a longest shortest path between vertices of the respective type contain at least one vertex in~$K$.
	Observe that since each connected component has diameter at most two, each third vertex in any longest shortest path must be in~$K$.
	Thus a shortest~$y$-$z$--path contains at least one vertex~$x_i\in K$ with~$\dis(x_i,y) < 3$.
	By definition, each vertex with the same type as~$y$ has the same distance to~$x_i$ and therefore the same distance to~$z$ unless there is no shortest path from it to~$z$ that passes through~$x_i$, that is, it is in the same connected component as~$z$.
	Thus, we can choose two arbitrary vertices of the respective types in different connected components.
	Observe that when computing the types, one can also precompute the connected components the corresponding vertices are in.
	Thus, checking whether there are two vertices of the respective type in different connected components is just a table lookup.
	We can compute the distance between~$y$ and~$z$ in~$O(k)$ time by computing~$\min_{x\in K}\dis(y,x) + \dis(x,z)$.
	Observe that the shortest path from~$y$ to~$z$ contains~$x_i$ and therefore~${\dis(y,x_i) + \dis(x_i,z) = \dis(y,z)}$.
	We can compute the diameter of~$G$ this way in~${O(k\cdot (n \log n+m + 2^{4k}))}$ time. 
\end{proof}

Note that the algorithm described in the above proof does not verify if~$K$ is indeed a vertex set such that each connected component in~$G-K$ has diameter at most two.
Indeed, even in the unweighted case to distinguish diameter two and three in $O(n^{2-\varepsilon})$, $\varepsilon > 0$, time would refute the SETH~\cite{AWW16}.
Thus, the above algorithm cannot efficiently verify if the input meets the stated conditions.
Hence, when using \cref{lem:dist-to-2-club}, we need a way to ensure this condition.

Recall that a cograph does not contain a $P_4$ as an induced subgraph.
Thus, any unweighted cograph has diameter at most two (but not every diameter-two graph is a cograph, consider e.\,g.\ a cycle on five vertices).
Moreover, given a graph~$G$ one can determine in linear time whether~$G$ is a cograph and can return an induced~$P_4$ if this is not the case~\cite{BCHP08,CPS85}.
This implies that in~$O(k\cdot (n+m))$ time one can compute a set~$K\subseteq V$ with~$|K|\leq 4k$ such that~$G - K$ is a cograph: 
Iteratively add all four vertices of a returned~$P_4$ into the solution set and delete those vertices from~$G$ until it is~$P_4$-free.
Thus, we can compute a set~$K$ that satisfy the conditions of \cref{lem:dist-to-2-club} and the following theorem is immediate.

\begin{theorem}%
	\label{thm:cograph}
	\di{} can be solved in~$O(k \cdot (n + m + 2^{16k}))$ time when parameterized by \param{distance~$k$ to cographs}.
\end{theorem}

\begin{proof}
	Let~$G=(V,E)$ be the input graph with \param{distance~$k$ to cograph}.
	Let~$K$ be a set of vertices such that~$G' = G - K$ is a cograph with~$|K| \leq 4k$.
	Recall that~$K$ can be computed in~$O(k \cdot (n+m))$ time.
	
	Thus, applying \cref{lem:dist-to-2-club} yields a running time of~$O(k \cdot (n + m + 2^{16k}))$.
	Note that since we are in the unweighted setting, we can replace Dijkstra's algorithm in the proof of \cref{lem:dist-to-2-club} by a simple breadth-first search and thus get rid of the log-factor in the running time.
\end{proof}

Note that a clique is also a cograph. 
Thus, following the same argumentation given after \cref{obs:dist-to-clique}, it follows that a generalization of \cref{thm:cograph} to the weighted case would significantly improve the state-of-the-art algorithm for \di{}.

\paragraph{Feedback edge number.}
We will prove that \wdi{} parameterized by \param{feedback edge number}~$k$ can be solved in~$O(k\cdot n \log n)$ time.
One can compute a minimum feedback edge set~$K$ (with $|K| = k$) in linear time by taking all edges not in a spanning tree.
Recently, this parameter was used to speed up algorithms computing maximum matchings~\cite{KNNZ18}.
Note that~$k \le m$, thus the subsequently provided $O(k\cdot n \log n)$-time algorithm is adaptive, that is, it is not slower than the standard~$O(n(n \log n + m))$-time algorithm but can be much faster in case~$k = o(m)$.
In the remainder of this section we will prove the following.

\begin{theorem}%
	\label{thm:fes}
	\wdi{} parameterized by \param{feedback edge number}~$k$ can be solved in~$O(k\cdot n \log n)$ time.
\end{theorem}
{
The algorithm behind the above theorem works roughly in two steps: 
In a first step, we apply data reduction rules.
On the one hand, these rules can shrink the graph considerably. 
On the other hand, these rules also create a special structure: 
After these rules are exhaustively applied, there are ``few'' vertices of degree at least three; moreover, these high-degree vertices are connected via ``few'' paths.
In the second step, the algorithm uses this structure in a case distinction to compute the diameter in~$O(k\cdot n \log n)$ time.

Note that the data reduction rules delete vertices from the graph.
However, since at the time of deletion, we do not know whether these vertices are contained in a shortest path defining the diameter, we need to keep additional information.
In particular, we introduce a second weight function~$\pen$ (for pending) and an integer~$s$.
Intuitively, $\pen(v)$ stores the length of a longest shortest path~$P$ with one endpoint being~$v$ and the other endpoint in~$P$ being already deleted by the data reduction rules.
The role of~$s$ is to store the length of a longest shortest path where \emph{both} endpoints are already deleted.
This leads to the following formal problem definition:

\problemdef{\dwdi}
{An undirected, connected graph~$G=(V,E)$, weight functions~${\w \colon E \rightarrow \N^+}$ and~$\pen \colon V \rightarrow \N$, and~$s \in \N$.}
{Compute $\max\{\dia^{\pen}(G), s\}$, where $$\dia^{\pen}(G) := \max_{v,w\in V} \{\dis_G^{\pen}(v,w)\} := \max_{v,w\in V} \{\pen(v) + \dis_G(v,w) + \pen(w)\}.$$}

Notice that if all $\pen$-weights and~$s$ are set to~$0$, then the problem is the same as \wdi.
We therefore start with initializing all $\pen$-weights and~$s$ to~$0$ and applying our reduction rule that removes degree-one vertices from the graph.
The main idea of the reduction rule is simple: 
If a degree-one vertex~$u$ is removed, then the value~$\pen(v)$ ($v$ is the unique neighbor of~$u$) is adjusted and we store in an additional variable~$s$ the length of a longest shortest path that cannot be recovered from the reduced graph.
This addresses the case that a longest shortest path has both its endpoints in pending trees (trees removed by our reduction rule) that are connected to the same vertex.
Initially, $s$ is set to zero.
The first reduction rule is defined as follows (see \cref{fig:fes-rr-example} for an example illustrating the subsequent two reduction rules).
\begin{figure}[t!]
	\centering
	\begin{tikzpicture}[]
		\begin{scope}[xshift = -6cm]
			\node[vertex,label=above:$v_1$] (v1) at (0,0) {};
			\node[vertex,label=above:$v_2$] (v2) at (2,0) {} edge node[label=1] {} (v1);
			\node[vertex,label=right:$v_3$] (v3) at (1,-1.5) {} edge node[midway,label=left:1] {} (v1)  edge node[midway,label=right:1] {} (v2);
			\node[vertex,label=right:$v_4$] (v4) at (0,-3) {} edge node[midway,label=left:3] {} (v3);
			\node[vertex,label=right:$v_5$] (v5) at (2,-3) {} edge node[midway,label=right:5] {} (v3);
			\node[] (s1) at (1,-2.7) {$s\ {=\ } 0$};
		\end{scope}
		
		\draw (-4.5,0.6) edge[bend left,->] node[midway,label=above:{\cref{rr:deg-one}}] {} (-2.5,0.6);
		\draw (-0.75,0.6) edge[bend left,->] node[midway,label=above:{\cref{rr:deg-one}}] {} (1.25,0.6);
		\draw (2.5,0.6) edge[bend left,->] node[midway,label=above:{\cref{rr:pc}}] {} (4.5,0.6);

		\begin{scope}[xshift = -2.5cm]
			\node[vertex,label=above:$v_1$] (w1) at (0,0) {};
			\node[vertex,label=above:$v_2$] (w2) at (2,0) {} edge node[label=1] {} (w1);
			\node[vertex,label=left:$\pen {=\ } 3$,label=right:$v_3$] (w3) at (1,-1.5) {} edge node[midway,label=left:1] {} (w1)  edge node[midway,label=right:1] {} (w2);
			\node[vertex,label=right:$v_5$] (w5) at (2,-3) {} edge node[midway,label=right:5] {} (w3);
			\node[] (s1) at (1,-2.7) {$s\ {=\ } 3$};
		\end{scope}

		\begin{scope}[xshift = 1cm]
			\node[vertex,label=above:$v_1$] (u1) at (0,0) {};
			\node[vertex,label=above:$v_2$] (u2) at (2,0) {} edge node[label=1] {} (u1);
			\node[vertex,label=left:$\pen {=\ } 5$,label=right:$v_3$] (u3) at (1,-1.5) {} edge node[midway,label=left:1] {} (u1)  edge node[midway,label=right:1] {} (u2);
			\node[] (s1) at (1,-2.7) {$s\ {=\ } 8$};
		\end{scope}

		\begin{scope}[xshift = 4cm]
			\node[vertex,label=left:$\pen {=\ } 5$,label=right:$v_3$] (t3) at (1,-1.5) {};
			\node[] (s1) at (1,-2.7) {$s\ {=\ } 8$};
		\end{scope}
					
		\tikzstyle{edge} = [color=black,opacity=.15,line cap=round, line join=round, line width=12pt]
		\begin{pgfonlayer}{background}
			\draw[edge] (v4.center) -- (v3.center) -- (v5.center);
			\draw[edge] (w3.center) -- (w5.center);
		\end{pgfonlayer}
	\end{tikzpicture}
	\caption{
		Example for the application of \cref{rr:deg-one,rr:pc}. 
		On the left is the input graph, middle left and middle right are the results of applying \cref{rr:deg-one}.
		On the right is the result of applying \cref{rr:pc} to the middle right graph.
		If no pen-value is displayed for a vertex~$v$, then~$\pen(v)=0$.
		The diameter-defining path is highlighted in the two left graphs and stored in~$s$ in the two right graphs (when the diameter-defining path is no longer contained in the remaining graph).
	}
	\label{fig:fes-rr-example}
\end{figure}

\begin{rrule}
	\label{rr:deg-one}
	Let~$u$ be a vertex of degree one and let~$v$ be its neighbor.
	Delete~$u$ and the incident edge from~$G$, set~$s = \max\{s, \pen(u) + \pen(v) + \w(\{u,v\})\}$ and~$\pen(v) = \max\{\pen(v), \pen(u) + \w(\{u,v\})\}$.
\end{rrule}

Before we analyze the running time and correctness, we first present a second reduction rule that we apply after \Cref{rr:deg-one} is not applicable anymore.
Since the resulting graph has no degree-one vertices we can partition the vertex set of the remaining graph into vertices~$V^{=2}$ of degree exactly two and vertices~$V^{\geq 3}$ of degree at least three.
Using standard argumentation we can show that~$|V^{\geq 3}| \in O(\min\{k,n\})$ and all vertices in~$V^{=2}$ are either in \emph{pending cycles} or in \emph{maximal paths}~\cite[Lemma 5]{BMKNN18}.
A maximal path is an induced subgraph~$P = x_0x_1\ldots x_a$ where~$\{x_i,x_{i+1}\} \in E$ for all~$0\leq i <a$, $x_0, x_a \in V^{\geq 3}$, $x_i\in V^{=2}$ for all~$0<i<a$, and~$x_0 \neq x_a$.  
A pending cycle is basically the same except~$x_0 = x_a$ and~$\deg(x_0)$ may possibly be two.
The set~$\mathcal{C}$ of all pending cycles and~$\mathcal{P}$ of maximal paths can be computed in~$O(n+m)$ time~\cite[Lemma 6]{BMKNN18}.
The second reduction rule works similar to \Cref{rr:deg-one}, but instead of deleting degree-one vertices, it removes pending cycles.

\begin{rrule}
	\label{rr:pc}
	Let~$C = x_0x_1\ldots x_{a}$ be a pending cycle.
	Let~$x_k$ be the vertex that maximizes~$\pen(x_k) + \dis(x_0,x_k)$ in~$C$.
	Delete all vertices in~$C$ except for~$x_0$ (and all incident edges) from~$G$, set~$s = \max\{s, \dia^{\pen}(C)\}$ and~$\pen(x_0) = \max\{\pen(x_0), \pen(x_k) + \dis(x_0,x_k)\}$.
\end{rrule}

We now prove the correctness of these two data reduction rules.
That is, given an instance~$(G,\w,\pen,s)$ of \dwdi{} let~$(G',\w',\pen',s')$ be the instance created by applying a data reduction rule~$R$ once.
Then, $R$ is \emph{correct} if~$\max\{s,\dia^{\pen}(G)\} = \max\{s',\dia^{\pen}(G')\}$.

\begin{lemma}
	\label{lem:rr1-correctness}
	\Cref{rr:deg-one,rr:pc} are correct.
\end{lemma}

\begin{proof}
	Let~$(G = (V,E),\w,\pen,s)$ be the input instance of \dwdi{} and~$(G' = (V',E'),\w',\pen',s')$ the instance resulting of an application of \cref{rr:deg-one} to the degree-one vertex~$u$ with neighbor~$v$ or \cref{rr:pc} to a pending cycle~$C = x_0,x_1,\ldots,x_a$.
	We start with making some statements that are true for both reduction rules.

	We first show that~$\dia^{\pen}(G) \ge \dia^{\pen}(G')$, that is, the ($\pen$-adjusted) diameter in~$G$ is at least as large as in~$G'$.
	To this end, let~$w,w' \in V'$ such that~$\dia^{\pen}(G') = \pen'(w) + \dis_{G'}(w,w') + \pen'(w')$.
	Observe that if~$w \neq v$ and~$w' \neq v$ (for \cref{rr:deg-one}) or~$w \neq x_0$ and~$,w' \neq x_0$ (for \cref{rr:pc}), then
	$$\pen'(w) + \dis_{G'}(w,w') + \pen'(w') \leq \pen(w) + \dis_{G}(w,w') + \pen(w') \le \dia^{\pen}(G).$$
	Thus, it remains to consider the case that~$w' = v$ for \cref{rr:deg-one} and~$w' = x_0$ for \cref{rr:pc} (the cases~$w=v$ respectively~$w = x_0$ are completely analogous). 
	In the case of \cref{rr:deg-one} we have
	\begin{align*}
		& \pen'(w) + \dis_{G'}(w,w') + \pen'(w') \\
		={}&{}\pen(w) + \dis_{G}(w,v) + \max\{\pen(v), \w(\{u,v\}) + \pen(u)\} \le \dia^{\pen}(G).
	\end{align*}
	In the case of \cref{rr:pc} we have for the ``furthest'' vertex~$x_k$ from~$x_0$ in~$C$ that
		\begin{align*}
		& \pen'(w) + \dis_{G'}(w,w') + \pen'(w') \\
		={}&{}\pen(w) + \dis_{G}(w,x_0) + \max\{\pen(x_0), \dis(\{x_0,x_k\}) + \pen(x_k)\} \le \dia^{\pen}(G).
	\end{align*}	
	Thus, $\dia^{\pen}(G) \ge \dia^{\pen}(G')$.

	Next, observe that~$s \le s'$.
	Moreover, observe that if~$s \ge \dia^{\pen}(G)$, then we have~$\max\{s,\dia^{\pen}(G)\} = s = s' = \max\{s',\dia^{\pen}(G')\}$ since~$s' \ge s \ge \dia^{\pen}(G) \ge \dia^{\pen}(G')$.
	Thus, it remains to consider the case~$s < \dia^{\pen}(G)$ and, hence, to show that~$\dia^{\pen}(G) = \max\{s',\dia^{\pen}(G')\}$.
	
	We split this last part of the proof into two parts, where we first consider~\cref{rr:deg-one} and then consider \cref{rr:pc} in the second part.
	For the first part, let~$w,w' \in V$ such that~$\dia^{\pen}(G) = \pen(w) + \dis_{G}(w,w') + \pen(w')$.
	We make a case distinction on the size of~$\{w,w'\} \cap \{u,v\}$ (that is, whether~$w$ or~$w'$ are equal to~$v$ or~$u$).

	Case 1: $|\{w,w'\} \cap \{u,v\}| = 2$.
	Since~$s < \dia^{\pen}(G)$, we have by definition of~$s'$ that
	\begin{align*}
		\dia^{\pen}(G) = \pen(u) + \dis_{G}(u,v) + \pen(v) = \pen(u) + \w(\{u,v\}) + \pen(v) = s'.
	\end{align*}
	Since~$\dia^{\pen}(G') \le \dia^{\pen}(G)$, it follows that~$\max\{s',\dia^{\pen}(G')\} = s' = \dia^{\pen}(G)$.

	In the following two cases we assume that~$\dia^{\pen}(G) > \pen(u) + \dis_{G}(u,v) + \pen(v)$; otherwise we are in Case 1.
	Hence, it follows that also~$s' < \dia^{\pen}(G)$ since~$s < \dia^{\pen}(G)$.

	Case 2:~$|\{w,w'\} \cap \{u,v\}| = 1$.
	Thus, we need to show~$\dia^{\pen}(G') \ge \dia^{\pen}(G)$ (as we already proved~$\dia^{\pen}(G') \le \dia^{\pen}(G)$ and assume~$s' < \dia^{\pen}(G)$).
	To this end, let~$w' \in \{u,v\}$ and~$w \notin \{u,v\}$.
	Hence, we have
	\begin{align*}
		\dia^{\pen}(G) 	& = \pen(w) + \dis_{G}(w,w') + \pen(w') \\
						& = \pen(w) + \dis_{G}(w,v) + \max\{\pen(v), \pen(u) + \w(\{u,v\})\} \\
						& = \pen'(w) + \dis_{G'}(w,v) + \pen'(v) \le \dia^{\pen}(G').
	\end{align*}
	Thus, $\dia^{\pen}(G') = \dia^{\pen}(G)$.

	Case 3:~$|\{w,w'\} \cap \{u,v\}| = 0$.
	Again, we need to show~$\dia^{\pen}(G') \ge \dia^{\pen}(G)$.
	To this end, neither~$w$ nor~$w'$ are changed by \cref{rr:deg-one}.
	Thus,
	$$\pen(w) + \dis_{G}(w,w') + \pen(w') = \pen'(w) + \dis_{G'}(w,w') + \pen'(w') \le \dia^{\pen}(G').$$
	This finishes the last case and concludes the proof for \cref{rr:deg-one}.
	
	We continue with the proof for \cref{rr:pc}.
	To this end we consider two cases: Either~$s' > \dia^{\pen}(G')$ or~$s' \leq \dia^{\pen}(G')$.
	
	Case 1:~$s' \geq \dia^{\pen}(G')$.
	We show that~$s' = \dia^{\pen}(G)$.
	Since~$s' \geq \dia^{\pen}(G')$, there is no shortest path of length~$s'+1$ in~$G'$.
	Since~$G$ and~$G'$ only differ in~$C$, it suffices to show that there is a shortest path of length~$s'$ in~$G$ and that there is no longer path that starts in~$C$.
	By construction, there is a pair of vertices~$x_i, x_j$ in~$C$ such that~$\dis^{\pen}_{G}(x_i,x_j) = s'$.
	Now assume that there is a shortest path of length at least~$s'+1$ in~$G$ that starts in~$C$.
	By construction the path has to end outside of~$C$ as otherwise~$s'$ would be larger.
	Let~$v$ be the other endpoint of the path.
	Then,~$\dia^{\pen}(G') \geq \dis_{G'}^{\pen}(x_0,v) > s'$---a contradiction.
	
	Case 2:~$s' < \dia^{\pen}(G')$.
	We will show that~$\dia^{\pen}(G) \leq \dia^{\pen}(G')$.
	We first define~$V_C = \{x_0,x_1,\ldots,x_{a-1}\}$ to be the set of vertices in~$C$.
	Again, let~$w,w' \in V$ such that~$\dia^{\pen}(G) = \pen(w) + \dis_{G}(w,w') + \pen(w')$ and we make a case distinction on the size of~$\{w,w'\} \cap V_C$.
	
	Subcase 1:~$|\{w,w'\} \cap V_C| = 0$.
	Since~$G$ and~$G'$ only differ in~$C$, we have
	\begin{align*}
		\dia^{\pen}(G) 	& = \pen(w) + \dis_{G}(w,w') + \pen(w') \\
						& = \pen'(w) + \dis_{G'}(w,w') + \pen'(w') \le \dia^{\pen}(G').
	\end{align*}

	Subcase 2:~$|\{w,w'\} \cap V_C| = 2$.
	In this case by definition of~$s'$, we have that~$s' = \dia_G^{\pen} \geq \dia_{G'}^{\pen}$---a contradiction.

	Subcase 3:~$|\{w,w'\} \cap V_C| = 1$.
	We assume without loss of generality that~$w \notin V_C$ and~$w' \in V_C$.
	Then we have
	\begin{align*}
		\dia^{\pen}(G) 	& = \pen(w) + \dis_{G}(w,w') + \pen(w') \\
						& \leq \pen(w) + \dis_{G}(w,x_0) + \max\{\pen(x_0), \pen(w') + \dis_{G}(\{x_0,w'\})\} \\
						& = \pen'(w) + \dis_{G'}(w,x_0) + \pen'(x_0) \le \dia^{\pen}(G').
	\end{align*}
	This finishes the last case and concludes the proof.
\end{proof}

We now analyze the running time of~\Cref{rr:deg-one,rr:pc}.

\begin{lemma}
	\label{lem:rr2-correctness}
	Given a pending cycle~$C = x_0x_1\ldots x_{a}$, \Cref{rr:pc} can be applied in~$O(a)$ time.
\end{lemma}

\begin{proof}
First, in~$O(a)$ time we compute~$k$ such that~$\dis(x_k,x_0) + \pen(x_k)$ is maximized and if~$k\neq 0$, then we set~$s = \max\{s,\pen(x_0) + \dis(x_k,x_0) + \pen(x_k)\}$.
(For~$k=0$ we do not update~$s$.)
It remains to show how to compute~$\dia^{\pen}(C)$, the longest shortest path that starts and ends in~$C$.
To this end, we first compute the sum~$W$ of all edge-weights in~$C$, that is,~$W = \sum_{i=0}^{a-1} \w(\{x_i,x_{i+1}\})$.

Next we define two distance measures~$d_{\cl},d_{\cc}$ (for clockwise and counter-clockwise) such that 
\begin{align*}
		 d_{\cl}(x_i,x_j) ={}&{}\w(\{x_i,x_{i+1 \bmod a}\}) \\ 
						& + \w(\{x_{i+1 \bmod a},x_{i+2\bmod a}\}) + \ldots +\w(\{x_{j-1 \bmod a},x_j\}) && \text{ and}\\
		d_{\cc}(x_i,x_j) ={}&{}\w(\{x_i,x_{i-1 \bmod a}\}) \\
						& + \w(\{x_{i-1 \bmod a},x_{i-2\bmod a}\}) + \ldots +\w(\{x_{j+1 \bmod a},x_j\}).
\end{align*}
Note that~$d_{\cl}(x_i,x_j) + d_{\cc}(x_i,x_j) = W$ and~$d_{cl}(x_i,x_j) = d_{\cc} (x_j,x_i)$.

We provide a dynamic program that only considers ``clockwise'' shortest paths between~$x_{\ell}$ and~$x_{j}$, that is, paths of length~$\pen(x_{\ell}) + d_{\cl}(x_{\ell},x_{j}) + \pen(x_{j})$ that satisfy~$d_{\cl}(x_{\ell},x_{j}) \le d_{\cc}(x_{\ell},x_{j})$ (otherwise it is not a shortest path).
Observe that all ``counter-clockwise'' paths will be considered in the iteration where the role of~$x_{j}$ and~$x_{\ell}$ is swapped as~$d_{cc}(x_{\ell},x_{j}) = d_{\cl} (x_{\ell},x_{j})$.

The dynamic program uses a table~$T$ with~$a$ entries, where the~$\ell^\text{th}$ entry corresponds to~$x_\ell$ and the value stored in the entry is the vertex~$x_j$ furthest from~$x_\ell$, formally, 
$$x_j := \argmax_{x \in \{x_i \, \mid \, d_{\cl}(x_{\ell},x_{i}) \, \le \, d_{\cc}(x_{\ell},x_{i})\}} \{ \dis(x,x_\ell) \}.$$
For initialization, we start with computing~$T[x_0]$ by checking in~$O(a)$ time all vertices in~$C$.
Besides the table~$T$, the dynamic program has one more variable~$r$ storing the length of a longest shortest path found so far.
Initially, $r = \pen(x_0) + \dis(x_k,x_0) + \pen(x_k)$.

Given~$x_j = T[x_\ell]$ for some vertex~$x_\ell$ the dynamic program computes the furthest vertex~$x_{j'}$ from~$x_{\ell+1}$ and updates~$r$ if any longest shortest path from~$x_{\ell+1}$ is longer than~$r$.
Note that the furthest vertex~$x_{j'}$ from~$x_{\ell+1}$ is either the furthest vertex~$T[x_\ell] = x_j$ from~$x_\ell$ or some vertex~$x$ that is ignored by~$x_\ell$. 
The only possible vertices that are ignored by~$x_\ell$ but not by~$x_{\ell+1}$ are the vertices~$x$ with~$d_{\cl}(x_{\ell},x) > d_{\cc}(x_{\ell},x)$ and~$d_{\cl}(x_{\ell+1},x) \leq d_{\cc}(x_{\ell+1},x)$.
Thus, we can compute the furthest vertex from~$x_{\ell+1}$ in constant amortized time as follows:
We can compute the furthest vertex~$x_{j'}$ from~$x_{\ell+1}$ by iterating over the vertices~$x_{j+1 \bmod a}, x_{j+2 \bmod a}, \ldots$ and check whether
\begin{align*}
	& d_{\cc}(x_{\ell+1},x_{j+1\bmod a}) \\
	={}& d_{\cc}(x_{\ell},x_j) - \dis(x_{\ell},x_{\ell+1}) + \dis(x_{j \bmod a},x_{j+1 \bmod a}) \leq W/2.
\end{align*}
If this first check is met, then we compute the ``pen''-distance~$d_{\cc}(x_{\ell+1},x_{k+1 \bmod a}) + \pen(x_{\ell+1}) + \pen(x_{k+1 \bmod a})$.
If this is larger than~$r$, then we update~$r$ with this value (a longer shortest path was found).
We then continue with~$x_{k+2 \bmod a}$ and so on until the first check is not met anymore.

The whole pending cycle can be checked in~$O(a)$ time in this way and we can set~$s = \max\{s,r\}$.
\end{proof}

We now analyze the running time of both reduction rules.

\begin{lemma}
	\label{lem:rr-runtime}
	\Cref{rr:deg-one,rr:pc} can be exhaustively applied in~$O(n+m)$ time.
\end{lemma}

\begin{proof}
Notice that we can sort all vertices by their degree in linear time using bucket sort.
Applying \Cref{rr:deg-one} or~\cref{rr:pc} takes constant time per deleted vertex.
After applying a reduction rule, we adjust the degree of the remaining vertex (either the unique neighbor of a degree-one vertex or the high-degree vertex in a pending cycle) in constant time by moving it to the appropriate bucket.
Note that applying \cref{rr:pc} can lead to a new vertex of degree one and an application of \cref{rr:deg-one} can lead to two maximal paths merging to either a longer maximal path or a pending cycle.
Since these cases can be detected in constant time and each vertex is only removed once, the overall running time to apply \cref{rr:deg-one,rr:pc} exhaustively is in~$O(n + m)$.
\end{proof}

We now present the algorithm that computes the maximum~$\dis^{\pen}(u,v)$ over all pairs of remaining vertices~$u,v$ after applying \cref{rr:deg-one,rr:pc} exhaustively.
This algorithm distinguishes between three different cases:
The longest shortest path has at least one endpoint in~$V^{\geq 3}$ (Case 1), its two endpoints are in the same maximal path (Case 2), or its endpoints are in two different maximal paths (Case 3).
 
{ 
\begin{proof}[of \cref{thm:fes}]
Let~$G=(V,E)$ be the input graph with \param{feedback edge number}~$k$ and let~$K$ be a feedback edge set with~$|K|=k$.

\emph{Case 1:}
First we perform Dijkstra's algorithm from each vertex~$v\in V^{\geq 3}$ and store for each vertex~$u \in V \setminus\{v\}$ the distance~$\dis(v,u)$ and update~$s = \max\{s, \pen(v) + \pen(u) + \dis(v,u)\}$.
This way we find all shortest paths that start or end in a vertex in~$V^{\geq 3}$ (or a pendant tree connected to such a vertex).

\emph{Case 2:}
This case is similar to the case of pending cycles (see \cref{rr:pc}).
The only adjustment is the computation of the index that is considered by~$x_{\ell+1}$ but not by~$x_\ell$.
For a maximal path~$P=x_0x_1\ldots x_a$, we compute~$W= \sum_{i=0}^{a-1} \dis(x_i,x_{i+1})$ and check whether the distance ``within'' a path between two vertices~$x_i,x_j$ ($i < j$) is at most as large as~$\dis(x_i,x_0) + \dis(x_0,x_a) + \dis(x_a,x_j)$.

\emph{Case 3:}
We set~$V_P := \{x_1,x_2,\ldots,x_{a-1}\}$ and~$\overline{V}_P := V \setminus (V_P \cup \{x_0,x_a\}) = \{v_1, v_2, \ldots, v_{n-a-1}\}$.
In the last case we have that~$u$ is in a maximal path~$P = x_0x_1 \ldots x_a$ and~$v$ is outside~$P$, that is,~$u \in V_P$ and~$v \in \overline{V}_P$.
We present an algorithm that takes~$O(n \log n)$ time for each maximal path to compute the length of a longest shortest path of the specified type.
As there are~$O(k)$ such maximal paths~\mbox{\cite[Lemma 5]{BMKNN18}}, the overall running time is~$O(k \cdot n\log n)$.

The algorithm uses a length-$|\overline{V}_P|$ array~$D$ where the~$i^\text{th}$ entry is the distance difference of~$v_i \in \overline{V}_P$ to~$x_0$ and~$x_a$ respectively, formally, $D[i] := \dis_G(x_0,v_i) - \dis_G(x_a,v_i)$.
Note that for some vertex~$x_j$ in~$P$, there is a shortest~$x_j$-$v_i$-path leaving~$P$ via~$x_a$ if and only if~$\dis_P(x_j,x_a) - \dis_P(x_j,x_0) \le D[i]$.
Furthermore,~$D$ can be computed in~$O(n)$ time from the distances computed in Case 1.
The values~$\dis_P(x_j,x_a)$ and $\dis_P(x_j,x_0)$ can also be computed easily in~$O(n)$ time.

We use~$D$ in the following way: 
The algorithm sorts~$D$ in~$O(n \log n)$ time in non-increasing order (for ease of notation, we still assume that the~$i^\text{th}$ entry of~$D$ correspond to~$v_i$).
As a result, we have that if a shortest~$x_j$-$v_i$-path leaves~$P$ via~$x_a$, then so does every shortest~$x_j$-$v_{i'}$-path for every~$i' < i$.
Furthermore, since for any~$j' > j$ we have~$\dis_P(x_{j'},x_a) - \dis_P(x_{j'},x_0) < {\dis_P(x_{j},x_a) - \dis_P(x_{j},x_0) \le D[i]}$, we have that every shortest~$x_{j'}$-$v_{i}$-path goes via~$x_0$.
See \cref{fig:monoton} for an illustration of this monotonicity which is exploited in our subsequent algorithm.

\newcommand{\DValueExample}{
	\node[vertex, label=above:$x_0$] at(1,0) (x0) {};
	\node[vertex, label=above:$x_1$] at(2,.15) (x1) {};
	\node[vertex, label=above:$x_2$] at(3,.225) (x2) {};
	\node[vertex, label=above:$x_3$] at(4,.15) (x3) {};
	\node[vertex, label=above:$x_4$] at(5,0) (x4) {};

	\node[vertex, label=below:$v_1$] at(1,-2.5) (v1) {};
	\node[vertex, label=below:$v_2$] at(3,-2.75) (v2) {};
	\node[vertex, label=below:$v_3$] at(5,-2.5) (v3) {};

	\draw (x0) -- (x1);
	\draw (x1) -- (x2);
	\draw (x2) -- (x3);
	\draw (x3) -- (x4);

	\draw[dashed] (x0) -- node[near end,label=left:$13$] {} (v1);
	\draw[dashed] (x4) -- node[near start,label=left:$10$] {} (v1);
	\draw[dashed] (x0) -- node[midway,label=left:$10$] {} (v2);
	\draw[dashed] (x4) -- node[midway,label=right:$10$] {} (v2);
	\draw[dashed] (x0) -- node[near start,label=right:$10$] {} (v3);
	\draw[dashed] (x4) -- node[near end,label=right:$13$] {} (v3);
}

\begin{figure}
	\centering
	\begin{tikzpicture}
		\begin{scope}
			\DValueExample
			
			\tikzstyle{edge} = [color=black,opacity=.15,line cap=round, line join=round, line width=12pt]
			\begin{pgfonlayer}{background}
				\draw[edge] (x1.center) -- (x2.center) -- (x3.center) -- (x4.center) -- (v1.center);
			\end{pgfonlayer}
			\node[] at (3.2,-3.4) {$\dis_P(x_1,x_4) - \dis_P(x_1,x_0) = 2 \le D[1]$};
		\end{scope}

		\begin{scope}[xshift=7cm]
			\DValueExample
			
			\tikzstyle{edge} = [color=black,opacity=.15,line cap=round, line join=round, line width=12pt]
			\begin{pgfonlayer}{background}
				\draw[edge] (x2.center) -- (x3.center) -- (x4.center) -- (v1.center);
				\draw[edge] (x2.center) -- (x3.center) -- (x4.center) -- (v2.center);
			\end{pgfonlayer}
			\node[] at (2.8,-3.4) {$\dis_P(x_2,x_4) - \dis_P(x_2,x_0) = 0 \le D[2]$};
		\end{scope}
		
		\node[] at (6.5,-2.7) {$D = [3, 0, -3]$};
	\end{tikzpicture}
	\caption{
		Example demonstrating the monotonicities used in the proof of \cref{thm:fes}.
		All weights that are not displayed are 1 and all pen weights are 0. 
		Observe that only for~$i=1$ a shortest $x_1$-$v_i$-path goes over~$x_4$ (see highlighted path on the left).
		The fact that a shortest $x_1$-$v_i$-path goes over~$x_4$ if and only if~$\dis_P(x_1,x_4) - \dis_P(x_1,x_0) \le D[i]$ can also be seen in the example: $D[2] < \dis_P(x_1,x_4) - \dis_P(x_1,x_0) = 3 - 1 \le D[1]$.
		Exchanging~$x_1$ with~$x_2$ as starting point, results in more shortest $x_2$-$v_i$-paths going over~$x_4$ (see the highlighted paths on the right with~$x_2$ as starting point).
	}
	\label{fig:monoton}
\end{figure}

The algorithm handles two cases separately: One for computing the longest shortest~$x_j$-$v_i$-path, $x_j\in V_P$ and~$v_i \in \overline{V}_P$, that contains~$x_0$ and one for computing longest shortest~$x_j$-$v_i$-path containing~$x_a$.
As these two cases are completely symmetric, we will discuss only the latter case.
For brevity, let~$\dis_{\max}(x_j)$ be the length of a longest shortest path starting in~$x_j$, leaving~$P$ via~$x_a$, and ending in some~$v \in \overline{V}_P$.
Formally, $\dis_{\max}(x_j) = \max \{\dis^{\pen}(x_j,v_i) \mid v_i \in \overline{V}_P \land \dis_G(x_j,v_i) = \dis_P(x_j,x_a) + \dis_G(x_a,v_i)\}$.
Thus, the task is to compute~$\max_{j \in [a-1]}\{\dis_{\max}(x_j)\}$.
To this end, the algorithm computes~$\dis_{\max}(x_j)$ for all~$j$.

For the initialization, the algorithm computes the sorted array~$D$.
Moreover, it computes the largest number~$i_1 \in [n-a-1]$ such that~$\dis_G(x_1,v_{i_1}) = \dis_P(x_1,x_a) + \dis_G(x_a,v_{i_1})$.
If no such number exists, then set~$i_1 := 0$.
Furthermore, for each~$i \in [i_1]$ compute~$\dis^{\pen}(x_1,v_i) = \pen(v_i) + \dis_G(v_i,x_a) + \dis_P(x_a,x_1) + \pen(x_1)$ and store the maximum in a variable~$r$ ($r$ will be returned at the end of the algorithm).
Due to~$D$ being sorted, this initialization phase can be done in~$O(i_1)$ time.
Moreover, due to~$D$ being sorted, we have~$r = \dis_{\max}(x_1)$ as for all~$i' > i_1$ every shortest~$x_1$-$v_{i'}$-path leaves~$P$ via~$x_0$.
This completes the initialization.

Next, the algorithm computes for each~$j \in \{2,3,\ldots,a-1\}$ the value~$\dis_{\max}(x_j)$.
Notice that~$\dis_{\max}(x_1)$ was computed in the initialization.
For~$j>1$ the algorithm is as follows:
Compute the largest number~$i_j \in [n-a-1]$ such that~$\dis_G(x_j,v_{i_j}) = \dis_P(x_j,x_a) + \dis_G(x_a,v_{i_j})$.
Note that due to the sorting of~$D$ we have that~$i_j \ge i_{j-1}$.
Hence, we find~$i_j$ in~$O(i_j - i_{j-1})$ time by simply start checking~$D$ at positions~$i_{j-1}+1, i_{j-1}+2, \ldots, i_{j}, i_{j}+1$ (note that, by definition of~$i_j$, the last check at position~$i_{j}+1$ fails).
For each~$i \in \{i_{j-1}+1,i_{j-1}+2,\ldots,i_{j}\}$ we do the following:
We first compute~$\dis^{\pen}(x_j,v_i) = \pen(v_i) + \dis_G(v_i,x_a) + \dis_P(x_a,x_j) + \pen(x_j)$ and store the maximum in a variable~$r'$.
We then update~$r$ with~$\max\{ r', r - \pen(x_{j-1}) + \pen(x_{j}) - \w(\{x_{j-1},x_j\}) \}$.
Observe that~$r = \dis_{\max}(x_j)$ as for~$v_i$ with~$i \in \{i_{j-1}+1,i_{j-1}+2,\ldots,i_{j}\}$ the algorithm computed~$\dis^{\pen}(x_j,v_i)$.
For all~$i \in [i_{j-1}]$ we know that all~$x_{j-1}$-$v_i$-paths leave~$P$ via~$x_a$.
Thus, we can simply update their length by~$\pen(x_{j}) - \pen(x_{j-1}) - \w(\{x_{j-1},x_j\})$.

Altogether, the algorithm runs in~$O(k (n \log n + \sum_{i=1}^{a-1} (i_j - i_{j-1}))) = O(kn \log n)$ time.
Combining this with \cref{lem:rr-runtime} concludes the proof of \cref{thm:fes}.
\end{proof}
}
}

\section{Parameters for Social Networks}
\label{sec:soc-networks}

Here, we study parameters that we expect to be small in social networks.
Recall that social networks have the ``small-world'' property and a power-law degree distribution~\cite{LH08,M67,New03,New10,NJ03}.  
The ``small-world'' property directly transfers to the \param{diameter}.
We capture the power-law degree distribution by the \param{$h$-index} as only few high-degree vertices exist in the network.
Thus, we investigate parameters related to the \param{diameter} and to the \param{$h$-index} starting with degree-related parameters. %

\subsection{Degree Related Parameters}

We next investigate the parameter \param{minimum degree}.
Unsurprisingly, the \param{minimum degree} is not helpful for parameterized algorithms.
In fact, we show that \di{} is $2$-GP-hard with respect to the combined parameter \param{bisection width and minimum degree}.
The \param{bisection width} of a graph~$G$ is the minimum number of edges to delete from~$G$ in order to partition~$G$ into two connected component whose number of vertices differ by at most one.

\begin{proposition}%
	\label{thm:BWisGPhard}
	\di{} is $2$-GP-hard with respect to \param{bisection width and minimum degree}.
\end{proposition}
\begin{proof}
Let $G=(V,E)$ be an arbitrary input graph for \di{} where $V = \{v_1,v_2,\ldots v_n\}$ and let~$d$ be the diameter of~$G$.
We construct a new graph $G'=(V',E')$ with diameter~$d+4$ as follows:
Let $V'= \{s_i, t_i, u_i \mid i \in [n]\} \cup \{w_i \mid i \in [3n]\}$ and $E'= T \cup W \cup E''$, where $T = \{\{s_i, t_i\}, \{t_i,u_i\} \mid i \in [n]\}, W = \{u_1, w_1\} \cup \{\{w_1, w_i\}\mid i \in ([3n]\setminus \{1\})\}$, and $E'' = \{\{u_i, u_j\}\mid \{v_i, v_j\} \in E\}$.

An example of this construction can be seen in \cref{fig:bisectiongp}.
\begin{figure}[t!]
	\centering
	\begin{tikzpicture}[scale=0.85]
		\node[vertex,label=right:$v_1$] (v1) at (-6,0) {};
		\node[vertex,label=right:$v_2$] (v2) at (-7,-1) {} edge (v1);
		\node[vertex,label=right:$v_3$] (v3) at (-6,-2) {} edge (v1) edge (v2);
		\node[vertex,label=right:$v_4$] (v4) at (-6,-3) {} edge (v3);

		\draw [->,decorate,decoration=snake] (-5,-1.5) -- (-4,-1.5);

		\node[vertex,label=right:$u_1$] (u1) at (0,0) {};
		\node[vertex,label=right:$u_2$] (u2) at (-1,-1) {} edge (u1);
		\node[vertex,label=right:$u_3$] (u3) at (0,-2) {} edge (u1) edge (u2);
		\node[vertex,label=right:$u_4$] (u4) at (0,-3) {} edge (u3);

		\node[vertex,label=above:$t_1$] (t1) at (-2,0) {} edge (u1);
		\node[vertex,label=above:$t_2$] (t2) at (-2,-1) {} edge (u2);
		\node[vertex,label=above:$t_3$] (t3) at (-2,-2) {} edge (u3);
		\node[vertex,label=above:$t_4$] (t4) at (-2,-3){} edge (u4);

		\node[vertex,label=above:$s_1$] (s1) at (-3,0) {} edge (t1);
		\node[vertex,label=above:$s_2$] (s2) at (-3,-1) {} edge (t2);
		\node[vertex,label=above:$s_3$] (s3) at (-3,-2) {} edge (t3);
		\node[vertex,label=above:$s_4$] (s4) at (-3,-3) {} edge (t4);;

		\def\radius{1.75}
		\def\n{12}
		\node[vertex,label=left:$w_1$] (w1) at (4,-1.5) {} edge (u1);
		\foreach \i in {2,...,\n} {
			\node[vertex] (w\i) at ({\radius * cos(360 * \i / \n - 15*(180 / \n)) + 4},{\radius * sin(360 * \i / \n - 15*(180 / \n)) - 1.5}) {} edge (w1);
			\node[] () at ({1.25 * \radius * cos(360 * \i / \n - 15*(180 / \n)) + 4},{1.25 *\radius * sin(360 * \i / \n - 15*(180 / \n)) - 1.5}) {$w_{\i}$};
		}

		\draw [decorate,decoration={brace,amplitude=10pt},xshift=-4pt,yshift=0pt] (2.2,1) -- (5.5,1);
		\node (W) at (3.75,2) {{$W$}};
		\draw [decorate,decoration={brace,amplitude=10pt},xshift=-4pt,yshift=0pt] (-3.2,1) -- (0.5,1);
		\node (W) at (-1.5,2) {{$T$}};
		\draw [decorate,decoration={brace,amplitude=10pt},xshift=-4pt,yshift=0pt] (0.6,-3.5) -- (-1.2,-3.5);
		\node (W) at (-0.5,-4.5) {{$E''$}};
		
		\tikzstyle{edge} = [color=black,opacity=.15,line cap=round, line join=round, line width=12pt]
			\begin{pgfonlayer}{background}
				\draw[edge] (v1.center) -- (v3.center) -- (v4.center);

				\draw[edge] (s4.center) -- (u4.center) -- (u3.center) -- (u1.center) -- (s1.center);
			\end{pgfonlayer}
	\end{tikzpicture}
	\caption{
		Example for the construction in the proof of \Cref{thm:BWisGPhard}. 
		The input graph given on the left side has diameter~$2$ and the constructed graph on the right side has diameter~$2+4 = 6$.
	}
	\label{fig:bisectiongp}
\end{figure}
We will now prove that all properties of \cref{def:k-GP-hard} hold.
It is easy to verify that the reduction runs in linear time and that there are~$6n$ vertices and~$5n + m $ edges in~$G'$.
Notice that~${\{s_i, t_i, u_i \mid i \in [n]\}}$ and~${\{w_i \mid i \in [3n]\}}$ are both of size~$3n$ and that there is only one edge ($\{u_1,w_1\}$) between these two sets of vertices.
The bisection width of~$G'$ is therefore one and the minimum degree is also one as~$s_1$ is only adjacent to~$t_1$.

It remains to show that~$G'$ has diameter~$d+4$.
First, notice that the subgraph of~$G'$ induced by~$\{u_i \mid i \in [n]\}$ is isomorphic to~$G$.
Note that~$\dis(s_i,u_i) = 2$ and thus~$\dis(s_i, s_j)= \dis(u_i, u_j) + 4 = \dis(v_i,v_j) + 4$ and therefore the diameter of~$G'$ is at least~$d+4$.
Third, notice that for all vertices~$x \in V' \setminus \{s_i\}$ it holds that~${\dis(s_i,x) > \dis(t_i,x)}$.
Lastly, observe that for all~$i \in [3n]$ and all vertices~${x \in V'}$ it holds that~${\dis(w_i, x) \leq \max\{\dis(s_1, x), 4\}}$.
Thus the longest shortest path in~$G'$ is between two vertices~$s_i,s_j$ and is of distance~$\dis(u_i,u_j) + 4 = \dis(v_i,v_j) + 4 \leq d + 4$.
\end{proof}

We mention in passing that the constructed graph in the proof of \cref{thm:BWisGPhard} contains the original graph as an induced subgraph and if the original graph is bipartite, then so is the constructed graph.
Thus, first applying the construction in the proof of \cref{thm:DtBisGPhard} (see also \cref{fig:bipartitegp}) and then the construction in the proof of \cref{thm:BWisGPhard} proves that \di{} is GP-hard even parameterized by the sum of \param{girth}, \param{bisection width}, \param{minimum degree}, and \dt{bipartite graphs}.

\subsection{Parameters related to both diameter and $h$-index}
\label{sec:comb}
Here, we will study combinations of two parameters where the first one is related to \param{diameter} and the second to \param{$h$-index} (see \Cref{fig:param-hierarchy} for an overview of closely related parameters).
We start with the combination \param{maximum degree and diameter}.
Interestingly, although the parameter is quite large, the naive algorithm behind \cref{obs:maxdeg-diam-alg} cannot be improved to a fully polynomial running time.

\begin{theorem}%
	\label{thm:maxdeg-diam}
	There is no~$(d+ \Delta)^{O(1)}(n+m)^{2-\epsilon}$-time algorithm that solves \di{} parameterized by \param{maximum degree~$\Delta$ and diameter~$d$} unless the SETH is false.
\end{theorem}
{
\begin{proof}
We prove a slightly stronger statement excluding~$2^{o(\sqrt[c]{d + \Delta})}\cdot (n+m)^{2-\epsilon}$-time algorithms for some constant~$c$.
Assume towards a contradiction that for each constant~$r$ there is a~$2^{o(\sqrt[r]{d + \Delta})}\cdot (n+m)^{2-\epsilon}$-time algorithm that solves \di{} parameterized by \param{maximum degree}~$\Delta$ and \param{diameter}~$d$.
\citet{ED16} have shown a reduction from \textsc{CNF-SAT} to \di{} where the resulting graph has maximum degree three such that for any constant~$\epsilon>0$ an~$O((n+m)^{2-\epsilon})$-time algorithm (for \di) would refute the SETH.
A closer look reveals that there is some constant~$c$ such that the diameter~$d$ in their constructed graph is in~$O(\log^c (n+m))$.
By assumption we can solve \di{} parameterized by \param{maximum degree} and \param{diameter} in $2^{o(\sqrt[c]{d + \Delta})}\cdot (n+m)^{2-\epsilon}$ time.
Observe that
\begin{align*}
&\ 2^{o(\sqrt[c]{d + \Delta})}\cdot (n+m)^{2-\epsilon} = 2^{o(\sqrt[c]{\log^c (n+m)})}\cdot (n+m)^{2-\epsilon}\\
=&\ (n+m)^{o(1)}\cdot (n+m)^{2-\epsilon} \subseteq O((n+m)^{2-\epsilon'})\text{ for some }~\varepsilon'>0.
\end{align*}
Since we constructed for some~$\epsilon'>0$ an~$O((n+m)^{2-\epsilon'})$-time algorithm for \di{} the SETH fails and thus we reached a contradiction. Finally, notice that ${(d + \Delta)^{O(1)} \subset 2^{o(\sqrt[c]{d + \Delta})}}$ for any constant~$c$. 
\end{proof}
}

\paragraph{$h$-index and diameter.}
We next investigate the combined parameter \param{$h$-index} and \param{diameter}.
The reduction by \citet{RW13} produces instances with constant \param{domination number} and logarithmic \param{vertex cover number} (in the input size).
Since the \param{diameter}~$d$ is linearly upper-bounded by the \param{domination number} and the \param{$h$-index} is linearly upper-bounded by the \param{vertex cover number}, any algorithm that solves \di{} parameterized by the combined parameter~$(d+h)$ in~$2^{o(d+h)}\cdot (n+m)^{2-\epsilon}$ time disproves the SETH.
We will now present an algorithm for \wdi{} parameterized by \param{$h$-index} and \param{diameter} that almost matches the lower bound.

\begin{theorem}%
	\label{thm:hind-diam}
	\di{} parameterized by \param{diameter}~$d$ and \param{$h$-Index}~$h$ is solvable in~$O(h \cdot (n \log n + m) + n \cdot d \cdot h \cdot  (d^h + h^d \log h))$ time.
\end{theorem}

\begin{proof}
Let~$H = \{x_1,\ldots,x_h\}$ be a set of vertices such that all vertices in~$V \setminus H$ have degree at most~$h$ in~$G$.
Clearly,~$H$ can be computed in linear time.
We will describe a two-phase algorithm with the following basic idea:
In the first phase it performs Dijkstra's algorithm from each vertex~$v \in H$, stores the distance to each other vertex and uses this to compute the ``type'' of each vertex, that is, a characterization by the distance to each vertex in~$H$.
In the second phase it iteratively increases a value~$e$ and verifies whether there is a vertex pair of distance at least~$e+1$.
If at any point no vertex pair is found, then the diameter of~$G$ is~$e$.

The first phase is straight forward:
Execute Dijkstra's algorithm from each vertex~$v$ in~$H$ and store the distance from~$v$ to every other vertex~$w$ in a table.
Then iterate over each vertex~$w \in V \setminus H$ and compute a vector of length~$h$ where the $i$th entry represents the distance from~$w$ to~$x_i$.
Also store the number of vertices of each type containing at least one vertex.
Since the distance to any vertex is at most~$d$, there are at most~$d^h$ different types.
This first phase takes~$O(h \cdot (m + n \log n))$ time.

For the second phase, we initialize~$e$ with the largest distance found so far, that is, the maximum value stored in the table and compute~$G' = G - H$.
Iteratively check whether there is a pair of vertices in~$V \setminus H$ of distance at least~$e+1$ as follows.
We check for each vertex~$v \in V \setminus H$ whether there are types such that no vertex of one of these types can be reached by a path of length at most~$e$ passing through a vertex in~$H$.
This can be done by computing the sum of the two type-vectors in~$O(h)$ time and comparing the minimum entry in this sum with~$e$.
If all entries are larger than~$e$, then no shortest path from~$v$ to some vertex~$w$ of the respective type of length at most~$e$ can contain any vertex in~$H$.
Thus we compute Dijkstra's algorithm from~$v$ in~$G'$ up to depth~$e$\footnote{By ``up to depth~$e$'' we mean that we run Dijkstra's algorithm with the addition that whenever the distance to a vertex is at least~$e$, then we do not add it to the stack (or priority queue) and if the distance is larger then~$e$, then we do not update its distance to the source. Similar as in the proof of \cref{obs:maxdeg-diam-alg}, we can show that the number of vertices and edges considered by the algorithm are at most~$h^e+1$ and~$h^e$, respectively.} and count the number of vertices of the respective types we found.
If these numbers are equal to the total number of vertices of the respective types, then for all vertices~$w$ of these type it holds that~$\dis(v,w) \leq e$.
If the respective numbers do not match, then there is a vertex pair of distance at least~$e+1$, and we can therefore increase~$e$ by one and start the process again.

There are at most~$d$ iterations in which~$e$ is increased and the check is done.
In each iteration, we have to compute the sum of type vectors for each vertex and perform Dijkstra's algorithm up to depth at most~$d$ in~$G'$.
Recall that the maximum degree in~$G'$ is~$h$ and therefore computing Dijkstra's algorithm up to depth~$d$ takes~$O(h^d \cdot d \cdot \log h)$ time.
Since~$\sum_{e=1}^d h^e < h^{d+1}$ for~$h \geq 2$, the overall running time is in~$O(h \cdot (n \log n + m) + n \cdot d \cdot h \cdot  (d^h + h^d \log h))$.
\end{proof}

\paragraph{Acyclic chromatic number and domination number.}
We next analyze the parameterized complexity of \di{} parameterized by \param{acyclic chromatic number}~$a$ and \param{domination number}~$d$.
The \param{acyclic chromatic number} of a graph is the minimum number of colors needed to color each vertex with one of the given colors such that each subgraph induced by all vertices of one color is an independent set and each subgraph induced by all vertices of two colors is acyclic.
The \param{acyclic chromatic number} upper-bounds the average degree, and therefore the standard~$O(n \cdot m)$-time algorithm runs in~$O(n^2 \cdot a)$ time.
We will show that this is essentially the best one can hope for as we can exclude~$f(a,d) \cdot (n+m)^{2-\varepsilon}$-time algorithms 
assuming SETH.
Our result is based on the reduction by \citet{RW13} and is modified such that the \param{acyclic chromatic number} and \param{domination number} are both four in the resulting graph.

\begin{theorem}%
	\label{thm:acn-ds}
	There is no~$f(a,d)\cdot (n+m)^{2-\epsilon}$-time algorithm for any computable function~$f$ that solves \di{} parameterized by \param{acyclic chromatic number}~$a$ and \param{domination number}~$d$ unless the SETH is false.
\end{theorem}
{
\begin{proof}
We provide a reduction from \textsc{CNF-SAT} to \di{} where the input instance has constant \param{acyclic chromatic number} and \param{domination number} and such that an~$O((n+m)^{2-\varepsilon})$-time algorithm refutes the SETH.
Since the idea is the same as in \citet{RW13} we refer the reader to their work for more details.
Let~$\phi$ be a \textsc{CNF-SAT} instance with variable set~$W$ and clause set~$C$.
Assume without loss of generality that~$|W|$ is even.
We construct an instance~$(G=(V,E), k)$ for \di{} as follows:

Randomly partition~$W$ into two set~$W_1,W_2$ of equal size.
Add three sets~$V_1,V_2$ and~$B$ of vertices to~$G$ where each vertex in~$V_1$ (in~$V_2$) represents one of~$2^{|W_1|} = 2^{|W_2|}$ possible assignments of the variables in~$W_1$ (in~$W_2$) and each vertex in~$B$ represents a clause in~$C$.
Clearly~$|V_1| + |V_2| = 2 \cdot 2^{|W|/2}$ and~$|B| = |C|$.
For each~${v_i\in V_1}$ and each~$u_j\in B$ we add a new vertex~$s_{ij}$ and the two edges~$\{v_i,s_{ij}\}$ and~$\{u_j,s_{ij}\}$ to~$G$ if the respective variable assignment does \emph{not} satisfy the respective clause.
We call the set of all these newly introduced vertices~$S_1$.
Now repeat the process for all vertices~$w_i\in V_2$ and all~$u_j$ in~$B$ and call the newly introduced vertices~$q_{ij}$ and the set~$S_2$.
Finally we add four new vertices~$t_1,t_2,t_3,t_4$ and the following sets of edges to~$G$:
$\{\{t_1,v\} \mid v\in V_1\}, \{\{t_2,s\} \mid s\in S_1\}, \{\{t_3,q\} \mid q\in S_2\}, \{\{t_4,w\} \mid w\in V_2\}, \{\{t_2,b\},\{t_3,b\} \mid b \in B\}$, and~$\{\{t_1,t_2\},\{t_2,t_3\},\{t_3,t_4\}\}$.
See \Cref{fig:acndom} for a schematic illustration of the construction. 

\begin{figure}[t]
	\centering
	\begin{tikzpicture}
	\node[vertex, color=colViolet,label=above:$u_1$] (u1) at (0,0) {};
	\node[vertex, color=colViolet,label=above:$u_2$] (u2) at (0,-1)  {};
	\node[vertex, color=colViolet,label=above:$u_3$] (u3) at (0,-2)  {};
	\node[] (udots) at (0,-3)  {\large{$\vdots$}};
	\node[vertex, color=colViolet,label=above:$u_m$] (un) at (0,-4)  {};

	\node[vertex, color=colGreen,label=above:$v_1$] (v1) at (-5,0) {};
	\node[vertex, color=colGreen,label=above:$v_2$] (v2) at (-5,-1)  {};
	\node[vertex, color=colGreen,label=above:$v_3$] (v3) at (-5,-2)  {};
	\node[vertex, color=colGreen,label=above:$v_4$] (v4) at (-5,-3)  {};
	\node[] (vdots) at (-5,-4)  {\large{$\vdots$}};
	\node[vertex, color=colGreen,label=above:$v_{2^{n/2}}$] (vn) at (-5,-5)  {};

	\node[vertex, color=colGreen,label=above:$w_1$] (w1) at (5,0) {};
	\node[vertex, color=colGreen,label=above:$w_2$] (w2) at (5,-1)  {};
	\node[vertex, color=colGreen,label=above:$w_3$] (w3) at (5,-2)  {};
	\node[vertex, color=colGreen,label=above:$w_4$] (w4) at (5,-3)  {};
	\node[] (wdots) at (5,-4)  {\large{$\vdots$}};
	\node[vertex, color=colGreen,label=above:$w_{2^{n/2}}$] (wn) at (5,-5)  {};

	\node[vertex, color=colRed,label=above:$s_{11}$] (s1) at (-2.5,0) {};
	\node[vertex, color=colRed,label=above:$s_{12}$] (s2) at (-2.5,-1)  {};
	\node[vertex, color=colRed,label=above:$s_{22}$] (s3) at (-2.5,-2)  {};
	\node[vertex, color=colRed,label=above:$s_{3m}$] (s4) at (-2.5,-3)  {};
	\node[vertex, color=colRed,label=above:$s_{41}$] (sn) at (-2.5,-4)  {};
	\node[] (sdots) at (-2.5,-5)  {\large{$\vdots$}};

	\node[vertex, color=colRed,label=above:$q_{1m}$] (q1) at (2.5,0) {};
	\node[vertex, color=colRed,label=above:$q_{22}$] (q2) at (2.5,-1)  {};
	\node[vertex, color=colRed,label=above:$q_{2m}$] (q3) at (2.5,-2)  {};
	\node[vertex, color=colRed,label=above:$q_{31}$] (q4) at (2.5,-3)  {};
	\node[vertex, color=colRed,label=above:$q_{43}$] (qn) at (2.5,-4)  {};
	\node[] (qdots) at (2.5,-5)  {\large{$\vdots$}};
	
	\node[ellipse, draw, minimum height=6cm, minimum width=1.2cm, color=colViolet] (B) at(0,-2) {};
	\node[color=colViolet] () at(0,1.5) {\large $B$};
	\node[ellipse, draw, minimum height=7cm, minimum width=1.2cm, color=colGreen] (V1) at(-5,-2.5) {};
	\node[color=colGreen] () at(-5,1.5) {\large $V_1$};
	\node[ellipse, draw, minimum height=7cm, minimum width=1.2cm, color=colGreen] (V2) at(5,-2.5) {};
	\node[color=colGreen] () at(5,1.5) {\large $V_2$};
	\node[ellipse, draw, minimum height=7cm, minimum width=1.2cm, color=colRed] (S1) at(-2.5,-2.5) {};
	\node[color=colRed] () at(-2.5,1.5) {\large $S_1$};
	\node[ellipse, draw, minimum height=7cm, minimum width=1.2cm, color=colRed] (S2) at(2.5,-2.5) {};
	\node[color=colRed] () at(2.5,1.5) {\large $S_2$};
	
	\node[vertex, color=colRed,label=below:$t_1$] at (-2.5,-6.4) (t1) {};
	\node[vertex, color=blue,label=below:$t_2$]   at (-1.25,-5.2) (t2) {};
	\node[vertex, color=brown,label=below:$t_3$]  at (1.25, -5.2) (t3) {};
	\node[vertex, color=colRed,label=below:$t_4$] at (2.5, -6.4) (t4) {};
	
	\draw (v1) to (s1);
	\draw (v1) to (s2);
	\draw (v2) to (s3);
	\draw (v3) to (s4);
	\draw (v4) to (sn);
	
	\draw (s1) to (u1);
	\draw (s2) to (u2);
	\draw (s3) to (u2);
	\draw (s4) to (un);
	\draw (sn) to (u1);
	
	\draw (un) to (q1);
	\draw (u2) to (q2);
	\draw (un) to (q3);
	\draw (u1) to (q4);
	\draw (u3) to (qn);
	
	\draw (q1) to (w1);
	\draw (q2) to (w2);
	\draw (q3) to (w2);
	\draw (q4) to (w3);
	\draw (qn) to (w4);
	
	\draw[dotted] (v1) to (t1);
	\draw[dotted] (v2) to (t1);
	\draw[dotted] (v3) to (t1);
	\draw[dotted] (v4) to (t1);
	\draw[dotted] (vn) to (t1);
	
	\draw[dotted] (s1) to (t2);
	\draw[dotted] (s2) to (t2);
	\draw[dotted] (s3) to (t2);
	\draw[dotted] (s4) to (t2);
	\draw[dotted] (sn) to (t2);
	
	\draw[dotted] (u1) to (t2);
	\draw[dotted] (u2) to (t2);
	\draw[dotted] (u3) to (t2);
	\draw[dotted] (un) to (t2);
	
	\draw[dotted] (u1) to (t3);
	\draw[dotted] (u2) to (t3);
	\draw[dotted] (u3) to (t3);
	\draw[dotted] (un) to (t3);
	
	\draw[dotted] (q1) to (t3);
	\draw[dotted] (q2) to (t3);
	\draw[dotted] (q3) to (t3);
	\draw[dotted] (q4) to (t3);
	\draw[dotted] (qn) to (t3);

	\draw[dotted] (w1) to (t4);
	\draw[dotted] (w2) to (t4);
	\draw[dotted] (w3) to (t4);
	\draw[dotted] (w4) to (t4);
	\draw[dotted] (wn) to (t4);
	
	\draw[dotted] (t1) to (t2);
	\draw[dotted] (t2) to (t3);
	\draw[dotted] (t3) to (t4);
	\end{tikzpicture}
	\caption{
		A schematic illustration of the construction in the proof of \Cref{thm:acn-ds}. 
		Note that the resulting graph has \param{acyclic chromatic number} five ($V_1 \cup V_2, B, S_1 \cup S_2 \cup \{t_1,t_4\}, \{t_2\}$ and~$\{t_3\}$, also represented by colors) and a \param{dominating number} four ($\{t_1,t_2,t_3,t_4\}$).
		}
	\label{fig:acndom}
\end{figure}

We will first show that~$\phi$ is satisfiable if and only if~$G$ has diameter five and then show that the \param{domination number} and \param{acyclic chromatic number} of~$G$ are five and four, respectively.
First assume that~$\phi$ is satisfiable.
Then, there exists some assignment~$\beta$ of the variables such that all clauses are satisfied, that is, the two assignments of~$\beta$ with respect to the variables in~$W_1$ and~$W_2$ satisfy all clauses.
Let~${v_1\in V_1}$ and~$v_2\in V_2$ be the vertices corresponding to~$\beta$. 
Thus for each~$b\in B$ we have~$\dis(v_1,b) + \dis(v_2,b) \geq 5$.
Observe that all paths from a vertex in~$V_1$ to a vertex in~$V_2$ that do not pass a vertex in~$B$ pass through~$t_2$ and~$t_3$.
Since all of these paths are of length~$5$, it follows that~$\dis(v_1,v_2) = 5$.
Observe that the diameter of~$G$ is at most five since each vertex is connected to some vertex in~$\{t_1,t_2,t_3,t_4\}$ and these four are of pairwise distance at most three.

Assume next that the diameter of~$G$ is five.
Clearly there is a shortest path between a vertex~$v_i\in V_1$ and~$v_j\in V_2$ of length five.
Thus there is no path of the form~$v_is_{ih}u_hq_{jh}w_j$ for any~$u_h \in B$.
This corresponds to the statement that the variable assignment of~$v_i$ and~$w_j$ satisfy all clauses and therefore~$\phi$ is satisfiable.

The domination number of~$G$ is four since~$\{t_1,t_2,t_3,t_4\}$ is a dominating set.
The acyclic chromatic number of~$G$ is at most five as~$V_1 \cup V_2, B, S_1 \cup S_2 \cup \{t_1,t_4\}, \{t_2\}$ and~$\{t_3\}$ each induce an independent set and each combination of them not including~$S_1 \cup S_2 \cup \{t_1,t_4\}$ only induce independent sets or stars.
Lastly, note that~$S_1 \cup S_2 \cup \{t_1,t_4\}$ and~$\{t_2\}$ or~$\{t_3\}$ only induces a star and an independent set,~$S_1 \cup S_2 \cup \{t_1,t_4\}$ and~$V_1 \cup V_2$ induces two trees of depth~$2$ (where~$t_1$ and~$t_4$ are the roots and~$S_1$ and~$S_2$ are the leaves), and~$S_1 \cup S_2 \cup \{t_1,t_4\}$ and $B$ induce a disjoint union of stars and isolated vertices as each vertex in~$S_1 \cup S_2 \cup \{t_1,t_4\}$ has maximum degree one in~$G[B \cup S_1 \cup S_2 \cup \{t_1,t_4\}]$.

Now assume that we have an~$O(f(k)\cdot (n+m)^{2-\epsilon})$-time algorithm for \di{} parameterized by \param{domination number} and \param{acyclic chromatic number}.
Since the constructed graph has~$O(2^{|W|/2} \cdot |C|)$ vertices and edges, this would imply an algorithm with running time
\begin{align*}
&\ O(f(9)\cdot (2^{|W|/2} \cdot |C|)^{2-\epsilon})\\
=&\ O(2^{(|W|/2) (2-\epsilon)} \cdot |C|^{(2-\epsilon)})\\
=&\ O(2^{|W| (1 - \epsilon / 2)} \cdot |C|^{(2-\epsilon)}) \\
=&\ 2^{|W| (1 - \epsilon')} \cdot (|C| + |W|)^{O(1)}\text{ for some }~\varepsilon'>0.
\end{align*}
Hence, such an algorithm for \di{} would refute the SETH.
\end{proof}
}

\section{Conclusion}
We have resolved the complexity status of \di{} for most of the parameters in the complexity landscape shown in \Cref{fig:param-hierarchy}.
However, several open questions remain.
For example, is there an $f(k)n^2$-time algorithm with respect to the parameter \param{diameter}?
Moreover, our algorithms working with parameter combinations have mostly impractical running times which, assuming SETH, cannot be improved by much.
So the question arises, whether there are parameters~$k_1, \ldots, k_\ell$ that allow for practically relevant running times like~$\prod_{i=1}^{\ell} k_i \cdot (n+m)$ or even~$(n+m) \cdot \sum_{i=1}^{\ell} k_i$?
The list of parameters displayed in \Cref{fig:param-hierarchy} is by no means exhaustive.
Hence, the question arises which other parameters are small in typical scenarios? 
For example, what is a good parameter capturing the special community structures of social networks~\cite{GN02a}?

\bibliographystyle{plainnat}
\bibliography{bib-diameter}

\end{document}